\documentclass[11pt]{article}

\usepackage[letterpaper,margin=1in]{geometry}

\usepackage{amsmath,amssymb,amsthm}
\usepackage{mathtools}
\usepackage{braket}

\usepackage{authblk}

\usepackage{cite}
\usepackage[colorlinks=true,
            linkcolor=blue,
            citecolor=blue,
            urlcolor=blue]{hyperref}
\usepackage[nameinlink,capitalize]{cleveref}

\newtheorem{theorem}{Theorem}[section]
\newtheorem{lemma}[theorem]{Lemma}
\newtheorem{proposition}[theorem]{Proposition}
\newtheorem{corollary}[theorem]{Corollary}

\newtheorem{observation}[theorem]{Observation}

\theoremstyle{definition}
\newtheorem{definition}[theorem]{Definition}

\theoremstyle{remark}
\newtheorem{remark}[theorem]{Remark}

\newcommand{\BQP}{\mathsf{BQP}}
\newcommand{\NP}{\mathsf{NP}}
\newcommand{\QMA}{\mathsf{QMA}}
\newcommand{\QIMA}{\mathsf{QIMA}}
\newcommand{\DQC}{\mathsf{DQC}}

\newcommand{\poly}{\operatorname{poly}}
\newcommand{\negl}{\operatorname{negl}}
\newcommand{\Tr}{\operatorname{Tr}}

\newcommand{\ketbra}[1]{\ket{#1}\!\bra{#1}}

\title{The Commuting Local Hamiltonian Problem: \\Relativized Evidence Against $\BQP$-Hardness}
\author{
Itay Shalit
\qquad\qquad
Mark Zhandry
\\
\large Stanford University
}
\begin{document}

\date{}
\maketitle

% \begin{abstract}
% The commuting local-Hamiltonian problem is a special case of the local-Hamiltonian problem, in which the terms of the Hamiltonian are required to pairwise-commute.  A long line of work has shown that the problem lies in
% $\NP$ for certain families of commuting local Hamiltonians.
% Nevertheless, there has been no formal evidence against the possibility that
% the general commuting local-Hamiltonian problem is
% $\QMA$-complete.

% To capture the complexity of commuting local Hamiltonians, Bostanci
% and Hwang introduced the class \(\QIMA\), defined by quantum
% verifiers whose local gates commute, such that the commuting
% local Hamiltonian problem is \(\QIMA\)-complete.

% In this work, we introduce a classical oracle model \(\QIMA^{\mathcal O}\), and prove an oracle separation
% \[
% \BQP^{\mathcal O}
% \not\subseteq
% \QIMA^{\mathcal O}.
% \]

% BQP is the class of languages decidable by an efficient quantum distinguisher (hence, $\BQP\subseteq\QMA$). Our separation is based on the \(2\)-Forrelation problem, which asks whether
% two Boolean functions \(f\) and \(g\) have large correlation between \(f\) and
% the Fourier transform of \(g\). While \(2\)-Forrelation is known to be
% solvable in \(\BQP^{\mathcal O}\), we prove that it is not
% in \(\QIMA^{\mathcal O}\).
% \end{abstract}

\begin{abstract}
The commuting local-Hamiltonian (CLH) problem is a restriction of the local-Hamiltonian problem, in which the terms of the Hamiltonian are required to pairwise-commute. A long line of work has shown that the problem lies in
$\NP$ for certain families of commuting local Hamiltonians.
Nevertheless, there has been no formal evidence against the possibility that
the general CLH problem is
$\QMA$-complete.

CLH is complete for the complexity class $\QIMA$, defined through quantum verifiers whose
local gates commute. Therefore, CLH
is $\QMA$-complete if and only if
$\QIMA=\QMA$. In this work, we introduce a classical-oracle analogue
$\QIMA^{\mathcal O}$ and construct a classical oracle
$\mathcal O$ such that
\[
    \BQP^{\mathcal O}
    \not\subseteq
    \QIMA^{\mathcal O}.
\]

Since
$\BQP^{\mathcal O}\subseteq\QMA^{\mathcal O}$ for any classical oracle $\mathcal{O}$, this implies $\QIMA^{\mathcal O}
    \neq
    \QMA^{\mathcal O}$ for our constructed oracle.
Thus, our result provides relativized evidence against the possibility that the general CLH problem is
$\BQP$-hard, and hence also against the possibility that it is QMA-complete.  
\end{abstract}

\clearpage

\tableofcontents

\clearpage

\section{Introduction}
\label{sec:introduction}

The local Hamiltonian problem asks whether the ground-state energy of a
given local Hamiltonian is at most \(\alpha\) or at least \(\beta\). Kitaev showed that under
the promise \(\beta-\alpha=\Theta( \frac{1}{\operatorname{poly}(n)})\), this problem is complete for \(\QMA\), the class of languages
decidable by an efficient quantum verifier given a quantum witness.

A central question in quantum complexity theory is how the complexity of the
problem changes under structural restrictions on the Hamiltonian. One of the
main open cases is the commuting local Hamiltonian problem, in which all local
terms are guaranteed to commute. By imposing commutativity on the local terms in the Hamiltonian, one removes noncommuting observables—one of the central sources of uniquely quantum behavior. It therefore provides a natural setting for isolating which quantum computational phenomena rely on noncommutativity, and which persist even in its absence.

Commuting local Hamiltonians are also a useful testing ground for broader questions in quantum many-body theory and quantum complexity. In particular, they have played a role in attempts to understand the quantum PCP conjecture: commuting projector Hamiltonians arise naturally from quantum error-correcting codes, and the proof of the NLTS conjecture is based on quantum LDPC code Hamiltonians~\cite{Anshu_2023}; conversely, structural results for commuting Hamiltonians have exposed obstacles to adapting classical PCP techniques to the quantum setting~\cite{aharonov2013commutinglocalhamiltonianlocallyexpanding}. Commuting Hamiltonians also provide a tractable setting for studying the entanglement structure of ground states and the possibility of area laws~\cite{Mehta2016BehaviorOO}, as well as thermal properties and efficient preparation of Gibbs states~\cite{kastoryano2016quantumgibbssamplerscommuting}. Thus, understanding commuting local Hamiltonians may shed light not only on the sources of quantum computational hardness, but also on some of the central structural questions concerning low-energy and thermal states of many-body quantum systems.

Research into the complexity of the commuting local-Hamiltonian problem began with the work of Bravyi and Vyalyi, who showed that for 2-local commuting Hamiltonians the problem is in NP~\cite{bravyi2004commutativeversionklocalhamiltonian}. It was followed by a series of works which extended this result to other families of commuting local-Hamiltonians~\cite{aharonov2011complexitycommutinglocalhamiltonians, schuch2011complexitycommutinghamiltonians, aharonov2013commutinglocalhamiltonianlocallyexpanding, Aharonov_two_dimensional_CLH, irani2023commutinglocalhamiltonianproblem, bostanci2025commutinglocalhamiltonians2d}. On the other hand, Gosset et al.\ showed that the ground-space connectivity problem, which is QCMA-complete for general local-Hamiltonians, remains so for commuting local-Hamiltonians~\cite{Gosset_2017}. Finally, despite the fact that the commuting local-Hamiltonian problem is in NP for limited families of Hamiltonians, there has been no formal evidence against the possibility that in its most general form the problem is QMA-complete. 

Bostanci and Hwang defined the complexity class \(\QIMA=\bigcup_{k\in\mathbb{N}}\QIMA_k\), for which the commuting local-Hamiltonian problem is complete.\footnote{Bostanci and Hwang prove that the commuting local-projector
problem, a special case of the commuting local-Hamiltonian problem,
is QIMA-complete~\cite{bostanci2025commutinglocalhamiltonians2d}. Completeness of the general commuting
local-Hamiltonian problem is implicit in their work; we provide
an explicit proof in Appendix~\ref{app:QIMA-complete}. We use an equivalent reflection
normal form of their verifier definition, which allows general
local unitaries; the equivalence is
explained in Section~\ref{subsec:reflection_requirement}.}

\begin{definition}[Informal definition of $\mathrm{QIMA}_k$]\footnote{See Section \ref{sec:preliminaries} for a formal definition.}
A promise problem is in $\mathrm{QIMA}_k$ if it admits a QMA verification procedure of the following restricted form. On input $x$, the verifier efficiently constructs a polynomial number of mutually commuting $k$-local reflections
\[
R_1,\ldots,R_m .
\]

It then receives a polynomial-size quantum witness\footnote{The reflections $R_1,\ldots,R_m$ may depend on the instance but not on the witness.}. To test a reflection $R_i$, the verifier applies the Hadamard test: it prepares a fresh control qubit in the state $\ket{+}$, applies controlled-$R_i$ to the witness, and measures the control qubit in the $X$ basis. The verifier accepts if every such measurement returns the $+$ outcome.

The class is defined with perfect completeness. That is, on a
YES-instance some witness is accepted with probability $1$,
whereas on a NO-instance, every witness is accepted with probability $0$.
\end{definition}
\begin{remark}
Defining the class with completeness $1$ and soundness $0$ is without loss of generality, as any completeness and soundness parameters $0\leq\beta<\alpha\leq1$ define the same class as above. We can write $R_i=2P_i-I$ for an orthogonal projector $P_i$. The acceptance operator of the verifier is \(\prod_i P_i\), an orthogonal projector due to commutativity, so the maximum acceptance probability over all witnesses is either \(0\) or \(1\). Consequently, positive completeness and soundness strictly below \(1\) already imply perfect completeness and zero soundness.
\end{remark}

\begin{remark}
Throughout, the commuting local-Hamiltonian problem refers
to the family of problems obtained by fixing any constant
locality $k$. Statements of $\QIMA$-completeness are understood
in this familywise sense: each fixed-locality problem belongs
to $\QIMA$, and every promise problem $L\in\QIMA$ admits a
polynomial-time reduction to $k$-local CLH for some constant
$k$ depending on $L$, but not on the input length.
We do not assert that a single fixed locality is complete
for all of $\QIMA$.
\end{remark}

\subsection{Our Contributions}

We define an oracle analogue of \(\QIMA\), denoted by $\QIMA^{\cal{O}}$, and prove that 

\[\BQP^{\cal{O}} \not\subseteq \QIMA^{\cal{O}}.\]
Informally, our definition of $\QIMA^{\cal{O}}$ is as follows. 

\begin{definition}[$\QIMA^{\mathcal O}$]\label{def:QIMA-O}
A promise problem $L^{\mathcal O}$ is in
$\QIMA^{\mathcal O}$ if it admits the following kind of verification procedure.

On input $x$, the verifier first runs a uniform deterministic
polynomial-time classical preprocessing algorithm, which may make
adaptive classical queries to $O$. Based on the answers, it produces a
polynomial-size set of ``units'' 
\[
W_1^{\mathcal O},\ldots,W_m^{\mathcal O}
\]
each unit being a unitary circuit acting on a polynomial-sized quantum witness. A unit potentially includes quantum oracle queries. The resulting units are required to satisfy the following conditions. 

\begin{enumerate}
    \item The overall number of oracle queries across the different units is polynomial.
    \item Every unit \(W_j^{\cal{O}}\) that contains at least one quantum query to \({\cal{O}}\) is
    an exact reflection:
    \[
    (W_j^{\cal{O}})^\dagger=W_j^{\cal{O}},
    \qquad
    (W_j^{\cal{O}})^2=I.
    \]
    Oracle-free units may be arbitrary unitaries.

    \item On every promised oracle instance, the units commute
    pairwise:
    \[
    [W_i^{\cal{O}},W_j^{\cal{O}}]=0
    \qquad
    \text{for all }i,j.
    \]
\end{enumerate}

The verifier tests each unit $W_j^{\mathcal O}$ by preparing a fresh control
qubit in the state $\ket{+}$, applying controlled-$W_j^{\mathcal O}$, and
measuring the control qubit in the $X$ basis. It accepts if every such
measurement returns the $+$ outcome. This mechanism is often referred to in the literature as the \textit{Hadamard test}.

The verifier is required to have an inverse-polynomial completeness--soundness gap: on a
YES-instance, some witness is accepted with probability at least $c(n)$,
whereas on a NO-instance, every witness is accepted with probability at most
$s(n)$, with
\[
c(n)-s(n)\geq \frac{1}{\operatorname{poly}(n)}.
\]
\end{definition}

\begin{remark}
Concretely, a unit $W_j^{\mathcal O}$ may consist of an arbitrary
polynomial-length sequence of unitary operations interleaved with quantum
queries to $\mathcal O$. Different operations may act on different subsets
of the witness register.
If the unit contains an oracle query, the overall unit must be a reflection;
oracle-free units may be arbitrary unitaries. The commutation requirement
applies to the completed units $W_1^{\mathcal O},\ldots,W_m^{\mathcal O}$, so the constituent operations of a single unit need not commute. We allow controlled oracle queries: the oracle acts on the query
register if all designated control qubits are in state $|1\rangle$,
and acts as the identity otherwise.
\end{remark}

To prove the oracle separation, we use the Forrelation problem, which was defined by Aaronson~\cite{aaronson2009bqppolynomialhierarchy}. The problem is, given oracle access to two Boolean functions $f,g$, determining whether $f$ is highly-correlated with $\hat{g}$ or not, where $\hat{g}$ is the result of applying discrete Fourier-Transform over $\mathbb{F}_2$ to $g$. Following is a formal definition of the problem.

\begin{definition}[Forrelation] \label{def:forrelation}

    Let \(N=2^n\), and let \(H_N\) denote the normalized Hadamard
matrix,
\[
(H_N)_{x,y}
=
\frac{(-1)^{x\cdot y}}{\sqrt{N}}.
\]
For \(f,g:\{0,1\}^n\to\{-1,+1\}\), define their Forrelation by
\[
\Phi(f,g)
:=
\frac{1}{N}f^{\mathsf T}H_Ng.
\]
For fixed constants \(0<\beta<\alpha\leq 1\), the Forrelation promise
problem is
\[
\begin{aligned}
\textsc{Yes}:&\qquad \Phi(f,g)\geq\alpha,\\
\textsc{No}:&\qquad |\Phi(f,g)|\leq\beta.
\end{aligned}
\]
\end{definition} 

Variants of Forrelation were used in multiple works for distinguishing the power of quantum computational models from that of classical ones \cite{RazBqpPH2018, girish2025forrelationextremallyhard}. Forrelation can be decided by a polynomial-time quantum algorithm making constantly many oracle queries. In contrast, we prove that every oracle verifier satisfying the conditions of Definition~\ref{def:QIMA-O} and deciding Forrelation for all promised pairs \((f,g)\) requires exponentially many queries.

\begin{theorem}[Informal] \label{thm:informal_lower_bound}
Any oracle verifier satisfying Definition~\ref{def:QIMA-O}
and deciding Forrelation for all promised pairs $(f,g)$ must satisfy
\[
    C(n)+T(n)\ge \beta 2^n-O(1)
\]
for all sufficiently large even $n$. Here $C(n)$ counts classical
preprocessing queries, $T(n)$ counts all quantum oracle queries,
and $\beta$ is the Forrelation promise parameter.
\end{theorem}

We then use diagonalization to construct a single classical oracle \(\mathcal O\), whose slices encode Forrelation instances, such that
\[
\mathrm{BQP}^{\mathcal O}\not\subseteq\mathrm{QIMA}^{\mathcal O}.
\]

Throughout, access to $(f,g)$ is defined as follows. 
\begin{definition}[Oracle access]\label{def:oracle-access}
For $f,g:\{0,1\}^n\to\{-1,+1\}$, define the combined phase oracle by
\[
O_{f,g}\ket{b,x}
=
\begin{cases}
f(x)\ket{0,x}, & b=0,\\
g(x)\ket{1,x}, & b=1.
\end{cases}
\]
Quantum queries apply $O_{f,g}$, possibly conditioned on designated
control qubits all being in state $\ket{1}$.
The query register consists of an \textbf{oracle-selector qubit} \(b\), which selects between \(f\) and \(g\), and an \(n\)-qubit \textbf{address register} \(x\), which specifies the input to the selected function. We refer to the latter \(n\) qubits as the \textbf{address qubits}. A classical query returns either $f(x)$ or $g(x)$ at a chosen input $x$. Note that quantum queries to controlled $O_{f,g}$ gates allow for evaluating $f$ and $g$.
\end{definition}

Throughout the paper, $\QIMA^{O}$ refers to the class defined in
Definition~\ref{def:QIMA-O}. To address alternative choices of
relativization, in Section~\ref{sec:separation} we extend our separation to two natural
relaxations of this definition. In each case, we show in Section~\ref{section:comp_model} that a further weakening of the corresponding
restriction collapses the resulting class to $\QMA^{\mathcal{O}}$,
providing evidence that such a broader relaxation is too permissive
as an oracle analogue of $\QIMA$. Below, we explain the choices
in Definition~\ref{def:QIMA-O}, the extensions covered by our
lower bounds, and the relaxations that recover the full power
of $\QMA^{\mathcal{O}}$.

\subsubsection{The Reflection Requirement}\label{subsec:reflection_requirement}

Note that, for unrestricted circuit classes such as $\mathrm{BQP}$ and
$\mathrm{QMA}$, relativization has a fairly canonical meaning:
one allows the computation to make oracle queries. For $\QIMA$,
however, the verifier must consist of mutually commuting
units. Allowing oracle queries inside these units makes them non-local, and moreover results in equivalent formulations of the unrelativized class yielding \emph{in}equivalent oracle models. 

As a consequence, we make non-trivial choices in our definition of relativized $\QIMA$, which we now discuss. On one hand, our notion is very permissive, in that we allow our units to be highly non-local and make an arbitrary polynomial number of oracle queries. This permissiveness makes our oracle separation stronger. On the other hand, we require that units making oracle queries are reflections, which is a non-trivial restriction. We now discuss these choices.

Our formulation of the unrelativized class $\QIMA$ requires its
local units to be reflections. Under perfect completeness, we observe that reflecting units is without loss of generality. To see this, let $U_1,\ldots,U_m$
be arbitrary commuting local unitaries, and let $P_i$ project onto
the $+1$ eigenspace of $U_i$. Locality and the explicit descriptions
of the units allow us to efficiently construct the local projectors $P_i$.
Moreover, these projectors commute, and
\[
    P_i \preceq E(U_i),
    \qquad
    E(U_i) := \frac{2I+U_i+U_i^\dagger}{4},
\]
where $E(U_i)$ is the acceptance operator of the Hadamard test. That is, the probability of the Hadamard test on $U_i$ accepting an input state $\ket{\psi}$, is $\bra{\psi}E(U_i)\ket{\psi}$.

Replacing each $U_i$ by the reflection $2P_i-I$ therefore preserves
perfect completeness and can only improve soundness. This gives
a projective normal form for $\QIMA$. This projective normal form is then the basis for our relativization.

Now we observe that equivalent formulations of $\QIMA$ yield different classes once relativized. Indeed, for an oracle-containing unit given implicitly by a circuit, the
projector onto its $+1$ eigenspace need not have an efficient
implementation with comparable query complexity. Thus, relativizing the non-reflection version of $\QIMA$, despite being equivalent to the reflection version, yields a distinct class. Our Definition~\ref{def:QIMA-O} makes the explicit choice of having every oracle-containing unit being a reflection $R_i^{\mathcal{O}}$, which tests the projective constraint
\[
    \Pi_i^{\mathcal{O}} = \frac{I+R_i^{\mathcal{O}}}{2}.
\]
Oracle-free
units may still be arbitrary unitaries.

It is natural to consider different relativizations of $\QIMA$. For example, what if we directly relative the non-reflection version of $\QIMA$, allowing arbitrary commuting units, which may or may not be reflections. We prove in Section \ref{subsec:inverse-poly-reflection} that modifying Definition~\ref{def:QIMA-O} to allow for general unitary oracle-querying units results in a class equivalent to $\QMA^{\mathcal{O}}$. Moreover, this equivalence holds even if the units are restricted to being ``close'' to reflections, where ``close'' means inverse-polynomial distance in operator norm. In other words, unlike in the plain model, requiring the verifier units to be reflections is not without loss of generality in the oracle model. 

\begin{theorem}[Informal] \label{thm:informal_deviating_from_reflection}
For every efficiently computable inverse-polynomial tolerance
$\delta(n)>0$, every problem in $\QMA^{\mathcal{O}}$ admits an
efficient verifier consisting of a single Hadamard test of a
unitary $W^{\mathcal{O}}$ satisfying
\[
    \bigl\|W^{\mathcal{O}}-R\bigr\|_{\mathrm{op}}
    \le \delta(n),
\]
where $R$ is an efficiently implementable oracle-free reflection.
The verifier has an inverse-polynomial completeness--soundness
gap and no trusted workspace beyond the Hadamard-test control.
\end{theorem}
Theorem~\ref{thm:informal_deviating_from_reflection} demonstrates that restricting to reflecting units is necessary in order for the class to be a non-trivial restriction of $\QMA^{\mathcal{O}}$, provided the oracle-querying units are otherwise un-restricted.

\medskip

To give further evidence for a relativized separation between $\QIMA$ and $\BQP$, we also consider other potential variations of $\QIMA^{\mathcal{O}}$ where oracle-querying units are not required to be reflections. Despite Theorem \ref{thm:informal_deviating_from_reflection}, our lower bound does extend to
certain deviations from exact reflections. In
Section~\ref{sec:relax_reflection}, consider the setting which allows each oracle-containing
unit to be negligibly close in operator norm to an exact companion
reflection. We observe that this relaxation does not change our complexity class $\QIMA^{\mathcal{O}}$.

\begin{observation}[Informal]
Consider a relaxation of Definition~\ref{def:QIMA-O} with the
following changes. Each oracle-containing unit need only be
negligibly close in operator norm to an exact companion reflection,
uniformly over promised instances. The same preprocessing must
output these companions without additional classical queries.
Each companion must use no more quantum queries than the
corresponding unit, and the companions must commute with one
another and with all oracle-free units on promised instances. Then, for every fixed classical oracle $O$, this relaxation
defines the same class $\mathsf{QIMA}^O$ as Definition~1.4.
\end{observation}

In particular, our query lower-bound and oracle separation also hold for this
formulation. This negligible tolerance is essentially sharp: due to Theorem~\ref{thm:informal_deviating_from_reflection}, allowing
inverse-polynomial deviations already permits simulation of
$\QMA^{\mathcal{O}}$ verifiers. The simulation uses a single verifier unit,
but that unit may contain polynomially many oracle queries.

\medskip

A different variation is to consider units that are not close to reflections, but instead restrict how the unit access the oracle, say by limiting the number of queries. What restrictions allow us to decide Forrelation without collapsing the class to $\QMA^{\cal{O}}$?\footnote{Regardless of the answer to this question, we argue that the reflection model is a natural relativization of $\QIMA$. However, for completeness, we make an effort to address the question.} A unit containing a single oracle query has the form
\[
    U^{\mathcal{O}} = A O^{(Q)} B,
\]
where $O^{(Q)}$ denotes an oracle query on register $Q$, and
$A,B$ are oracle-free polynomial-size circuits. For Forrelation,
the query is $O_{f,g}^{(Q)}$. $A$ and $B$ may act
on the full witness register, including all address qubits,
and the completed unit need not be a reflection.
Section~\ref{subsec:one-query-collapse} proves the following.

\begin{theorem}[Informal]\label{thm:informal-single-unit-collapse}
The variant of Definition~\ref{def:QIMA-O} in which every
oracle-containing unit has the one-query form above, without
the reflection requirement, equals $\QMA^{\mathcal{O}}$.
Moreover, every problem in $\QMA^{\mathcal{O}}$ can be verified
efficiently using just one such unit and therefore one oracle
query. The verifier has an inverse-polynomial
completeness--soundness gap, and all registers other than the
Hadamard-test control are supplied by the witness.
\end{theorem}

The two collapse results establish complementary limitations.
Theorem~\ref{thm:informal_deviating_from_reflection} achieves simulation
within any prescribed inverse-polynomial distance of an exact
reflection, while Theorem~\ref{thm:informal-single-unit-collapse} achieves
simulation with a single oracle query. The one-query construction in Theorem~\ref{thm:informal-single-unit-collapse}
does not preserve the near-reflection guarantee, so it does not
subsume Theorem~\ref{thm:informal_deviating_from_reflection}.

In Section~\ref{subsec:restricted-address-support-lower-bound}, we
complement the one-query collapse by limiting the number of
address qubits on which the oracle-free circuits $A$ and $B$
may together act nontrivially.
We retain all units allowed by Definition~\ref{def:QIMA-O}
and additionally permit non-reflection units of the form
\[
    U_i^{f,g}=A_i O_{f,g}^{(Q_i)}B_i.
\]
For each additional unit, however, $A_i$ and $B_i$ must together
act nontrivially on at most $k$ of the query's $n$ address qubits.
Equivalently, there must be at least $n-k$ address qubits on which
both circuits act as the identity. These untouched qubits may
differ between units.

We call $k$ the \emph{address-support bound}. The query itself
still acts on the full address register, and the restriction
does not limit the action of $A_i,B_i$ on the oracle-selector
qubit or on other witness qubits. All completed units must
still commute on promised instances, and no trusted workspace
is allowed beyond the Hadamard-test controls.

\begin{theorem}[Informal]
Suppose every additional one-query unit has address-support
bound at most $k(n)$. Any verifier in this extended model
that decides Forrelation must satisfy, for even $n$,
\[
    C(n)+2^{k(n)}T(n)\ge \beta 2^n-O(1),
\]
where $C(n)$ counts classical preprocessing queries and $T(n)$
counts all quantum queries, including those in the original
reflection units and the additional one-query units.
Consequently, if $C(n)$ is polynomial, then
\[
    T(n)=\Omega\!\left(2^{n-k(n)}\right).
\]
In particular, polynomial-query verification is impossible
whenever $n-k(n)=\omega(\log n)$.
\end{theorem}

For $k(n)\le (1-\varepsilon)n$, where $\varepsilon>0$ is constant,
this gives an exponential query lower bound. More generally,
it gives a superpolynomial lower bound whenever the oracle-free
circuits surrounding each additional query leave a
superlogarithmic number of address qubits untouched.
By contrast, the collapse to $\QMA^{\mathcal{O}}$ in
Section~\ref{subsec:one-query-collapse} allows these circuits
to act on all $n$ address qubits. The two results thus compare
restricted and unrestricted address support within the same
one-query form, without identifying an exact threshold between
the two regimes.

\subsubsection{Ancilla Qubits}

A $\QIMA^{\mathcal O}$ verifier has no trusted ancilla qubits beyond the fresh control
qubits used for the individual Hadamard tests. Allowing even one
additional trusted active qubit gives a model equivalent to
$\mathrm{QMA}$ in the plain setting, and to
$\mathrm{QMA}^{\mathcal{O}}$ in the oracle setting, as shown in
Section~\ref{subsec:ancilla-collapse}.\footnote{In this reduction we only require $\QIMA$ to have an inverse-polynomial completeness-soundness gap. Without trusted ancillas this is equivalent to the
$(1,0)$ formulation of $\QIMA$, but this does not necessarily hold when a trusted ancilla is added to the model.} 
\begin{theorem}[Informal]
Consider the variants of $\QIMA$ and $\QIMA^{\mathcal{O}}$
that allow an inverse-polynomial completeness--soundness gap
and give the verifier one additional trusted qubit initialized
to $\ket{0}$, on which the commuting units may act.
This qubit is separate from the usual Hadamard-test controls.

These variants equal $\QMA$ and $\QMA^{\mathcal{O}}$,
respectively. Moreover, in the oracle setting, the simulation
can be realized using a single oracle-containing unit that
is an exact reflection.
\end{theorem}
\noindent This result motivates excluding trusted ancillas from our $\QIMA^{\mathcal O}$ model.

\subsubsection{Locality}
Locality of the verifier's units, which is required in $\QIMA$, cannot be retained after introducing oracle
access: a query to a function on $n$-bit strings acts on an $n$-qubit query
register. Thus, in $\QIMA^{\mathcal{O}}$, we allow a unit to act on any number of qubits. 

As a result, our oracle model need not recover the plain class \(\mathsf{QIMA}\) when the oracle is trivial. For comparison, in the QMA–QCMA separation of Bostanci et al. \cite{QMA_QCMA}, the relativized verifiers are uniform polynomial-size circuits consisting of elementary local gates and oracle gates. If the oracle query unitary is the identity, its gates can be deleted, recovering the corresponding unrelativized class. In our model, however, the completed commutation units may be nonlocal even when they contain no oracle queries. Deleting trivial oracle calls therefore need not produce a QIMA verifier. One reason is that decomposing the remaining units into local gates need not preserve commutativity. Indeed, our $\QMA^{\mathcal{O}}$ simulation in Section~\ref{subsec:inverse-poly-reflection} implies that $\QIMA^{O_0}=\QMA$ for the identity oracle $O_0$.

This can be understood as evidence for the strength of our $\QIMA^{\mathcal{O}}$ model.
Enlarging the model strengthens the lower-bound statement, which consequently applies to every narrower oracle extension contained in it. For example, one could restrict the commutation units to constant-locality oracle-free reflections and individual oracle-query reflections, requiring all units to commute. Our separation applies to this model as well. 

\subsubsection{Completeness and Soundness Parameters}

While $\QIMA$ is defined with perfect completeness and zero
soundness, Definition~\ref{def:QIMA-O} allows $\QIMA^{\mathcal O}$ verifiers to have any efficiently
computable completeness and soundness parameters separated by
an inverse-polynomial gap. Our separation therefore also applies
to the perfectly complete restriction of this oracle model.

When every tested unit is an exact reflection and the units commute,
the verifier's acceptance operator is an orthogonal projector.
Its maximum acceptance probability is therefore either zero or one.
This explains why positive completeness and soundness below one
suffice for the perfect-parameter formulation of plain $\QIMA$.

In our oracle model, however, oracle-free units may be arbitrary
unitaries, so the overall acceptance operator need not be a
projector. We do not establish whether requiring perfect
completeness changes this class. The collapse results in Section~\ref{section:comp_model} use the
inverse-polynomial-gap convention as well.

\subsection{Related Work}

\subsubsection{Commuting Local-Hamiltonians}
Bravyi and Vyalyi initiated the research into the commuting local-Hamiltonian problem, and showed that for 2-local commuting Hamiltonians the problem is in NP~\cite{bravyi2004commutativeversionklocalhamiltonian}. A series of works then extended their result by showing that it holds for additional families of commuting local-Hamiltonians. Aharonov and Eldar proved it for 3-local Hamiltonians over qubits, and for 3-local Hamiltonians over qutrits under a nearly-Euclidean geometric assumption~\cite{aharonov2011complexitycommutinglocalhamiltonians}. Schuch showed that commuting Hamiltonians on a two-dimensional square lattice of qubits are in NP~\cite{schuch2011complexitycommutinghamiltonians}. Aharonov and Eldar showed that sufficiently good local expansion of a commuting local-Hamiltonian's interaction graph
allows ground-energy approximation in NP~\cite{aharonov2013commutinglocalhamiltonianlocallyexpanding}. Aharonov, Kenneth, and Vigdorovich substantially generalized the two-dimensional qubit result to a broad class of nearly-Euclidean surface complexes, while additionally showing that the corresponding ground states can be efficiently prepared~\cite{Aharonov_two_dimensional_CLH}. Irani and Jiang extended the two-dimensional square-lattice result beyond qubits, proving NP containment for qutrits, as well as for factorized commuting Hamiltonians of arbitrary local dimension~\cite{irani2023commutinglocalhamiltonianproblem}. Most recently, Bostanci and Hwang proved NP containment for rank-one commuting Hamiltonians in two dimensions with unrestricted local dimension and locality, and obtained the first such result for a family of three-dimensional commuting Hamiltonians~\cite{bostanci2025commutinglocalhamiltonians2d}.

While the above results show that for certain families of commuting local-Hamiltonian, the local-Hamiltonian problem is in NP, Gosset, Mehta and Vidick proved that the ground state connectivity problem is as hard for commuting local-Hamiltonians as it is for general local-Hamiltonians~\cite{Gosset_2017}. Specifically, the problem is QCMA-hard in both cases.

\subsubsection{The Forrelation Problem}

Aaronson and Ambainis introduced Forrelation as a decision problem that
exhibits a separation between quantum and randomized query
complexity: Forrelation is solvable with a constant number of quantum queries,
whereas every bounded-error randomized algorithm requires
$\widetilde{\Omega}(\sqrt{N})$ queries~\cite{aaronson2014forrelationproblemoptimallyseparates}. A series of works then proved improved separations from extensions of this problem. Bansal and Sinha 
proved the conjectured optimal lower bound for $k$-Forrelation for every
$k$~\cite{BansalForrelation2021}, and Sherstov, Storozhenko, and Wu obtained optimal separations
between randomized and quantum query complexity~\cite{SherstovForrelation2021}.  Raz and Tal used a distributional
variant of Forrelation to construct an oracle relative to which
$\BQP$ is not contained in the Polynomial Hierarchy~\cite{RazBqpPH2018}.

More recently, Girish and Servedio studied the extremal version of the problem,
in which the Forrelation is promised to be either $+1$ or $-1$.  They proved an exponential
randomized-query lower bound even though the problem has an exact one-query
quantum algorithm~\cite{girish2025forrelationextremallyhard}.  Girish subsequently used Fourier-growth bounds
to show that Forrelation requires exponentially many queries in the
$\DQC_1$ model, and asked whether it can be solved by commuting IQP
circuits~\cite{GirishFourierSpectrum2026}.  Buzet and Chailloux answered this question positively:
the signed problem can be solved by one IQP circuit making one joint oracle
query, while the absolute-value problem can be solved by two IQP executions
together with efficient classical post-processing~\cite{buzet2026iqpcircuits2forrelation}.

Although IQP and QIMA both impose a commutativity restriction, the corresponding
computational models are substantially different.  An IQP computation begins
with a register prepared in the state $\lvert 0^m\rangle$,
measures that register at the end, and may accept according to an
efficiently recognizable set of outcomes.  In QIMA, by contrast, the target
register is supplied by an untrusted prover: apart from the control qubits used
for the individual Hadamard tests, the verifier has no trusted initialized
workspace. Indeed, in Section \ref{section:comp_model}, we prove that extending the QIMA model by allowing the verifier to use a trusted ancilla qubit, collapses the class to QMA. Thus the IQP algorithms of~\cite{buzet2026iqpcircuits2forrelation} do not yield QIMA
verifiers for Forrelation.

\subsection{Technical Overview}
\label{sec:technical-overview}

We first explain the Forrelation lower bound for
$\QIMA^{\mathcal{O}}$. The complementary lower bounds in
Sections \ref{subsec:restricted-address-support-lower-bound} and \ref{sec:relax_reflection} follow the same core argument, with
modifications intended to handle one-query units with restricted address
support and negligible deviations from exact reflections,
respectively. We then outline the approach taken for obtaining the collapse results in Section~\ref{section:comp_model}.

\subsubsection{The oracle separation.}
We prove the separation by establishing an exponential query lower bound for Forrelation against $\QIMA^{\mathcal O}$ verifiers. Let $C(n)$ be the number of classical oracle queries made during preprocessing, and let $T(n)$ be the total number of quantum oracle queries appearing in the units produced by the preprocessing algorithm. Our main technical statement is that, for every sufficiently large even $n$,
\[
    C(n)+T(n)\geq \beta 2^n-O(1).
\]
Thus, although Forrelation has a constant-query quantum algorithm, a $\QIMA^{\mathcal O}$ verifier for the problem needs an exponential number of queries. We first prove the lower bound when every oracle-containing unit is an exact reflection. At the end of the argument, we explain why it continues to hold when every unit with oracle queries is allowed to be within negligible distance of an exact reflection.

On a promised input, all verifier units have a common eigenbasis.
Each oracle-dependent reflection test accepts or rejects a common
eigenstate with certainty (as reflection eigenvalues are $\pm1$),
so multiplying their accepting projectors gives the projector onto
the subspace of states that pass all these tests. The oracle-free
units commute with every oracle-dependent projector, so they
preserve this common accepting subspace.
We can therefore consider the span of common eigenstates whose
acceptance probability exceeds a threshold between soundness and
completeness. An optimal witness can be chosen from the common
eigenbasis, so completeness guarantees that this span is nontrivial
on YES-instances. Soundness excludes every such eigenstate on
NO-instances, making the dimension exactly zero even when soundness
is nonzero.

The lower bound then follows from the polynomial
method~\cite{BealsPolynomialMethod1997,MinskyPerceptrons}. The crucial observation is that, once
the classical preprocessing is fixed, the dimension of the space described above
admits a polynomial expression in the oracle's truth-table entries,
with degree bounded by the number of quantum queries made by the verifier. Thus we obtain a
polynomial with \emph{exact zeros} on NO-instances. We choose a family of Forrelation instances on which this polynomial
is nonzero at a YES-instance and vanishes on many NO-instances.
Applying the standard Minsky--Papert symmetrization~\cite{MinskyPerceptrons}
yields a nonzero univariate polynomial of no greater degree with
many distinct roots. Counting these roots forces a large degree,
yielding the query lower bound.

\paragraph{Perfect Forrelation Pairs.} A perfect Forrelation pair is a pair of sign functions $g,h:\{0,1\}^n\rightarrow \{-1,1\}$ such that $\Phi(g,h)=1$, meaning that their Forrelation value is 1. Girish and Servedio established that $g,h$ are a perfect Forrelation pair if and only if $H_N g = h$, meaning that applying the Hadamard transform to the $2^n$-sized vector representing the truth table of $g$ results in the vector representing the truth table of $h$~\cite{girish2025forrelationextremallyhard}. Such functions are referred to in the literature as \textit{bent functions}, and they exist for any even $n$~\cite{Dillon1974}. For an explicit construction of a bent function pair, see Section \ref{subsec:perfect_forr_pairs}.

The useful feature of a bent function pair $(g,h)$ is that its Forrelation value changes in a particularly simple way when one of the functions is perturbed. If $f$ differs from $h$ on exactly $k$ truth-table entries, then
\[
    \Phi(g,f)
    =
    1-\frac{2k}{N}.
\]

This fact is proven in Lemma \ref{lemma:forrelation-under-hamming}. Thus, after fixing the first function to $g$, the Forrelation value $\Phi(g,f)$ depends only on the Hamming distance between $h$ and $f$. In particular, for every function $f$ at a distance $k$ from $h$ such that
\[
    \frac{1-\beta}{2}N
    \;\leq\;
    k
    \;\leq\;
    \frac{1+\beta}{2}N,
\]
$(g,f)$ is a no-instance.

\paragraph{Freezing the classical preprocessing.}
A $\QIMA^{\mathcal O}$ verifier may construct its quantum units based on a series of adaptive classical oracle queries. We remove this adaptivity by fixing the preprocessing computation. Run the deterministic preprocessing algorithm on a perfect Forrelation pair $(g,h)$, and let the "leaf" $\tau$ be the resulting transcript. Let
\[
    F\subseteq\{0,1\}^n
\]
be the set of positions of $h$ queried during this execution, and let $r=|F|$. Since the entire preprocessing makes at most $C(n)$ queries, $r\leq C(n)$.

Now keep $g$ fixed and consider every function $f$ satisfying
\[
    h(x)=f(x)
    \qquad
    \text{for every }x\in F.
\]
Every such $f$ produces exactly the same transcript $\tau$. Therefore, across this set of functions, the preprocessing algorithm produces the same quantum circuit descriptions.

\paragraph{Extracting an algebraic distinguisher.}
Next, we compile the verifier circuit into a function that must take distinct values for Yes and No instances.

Fix the transcript $\tau$. Write the oracle-containing units produced on this transcript as exact reflections
\[
    R_1^f,\ldots,R_a^f,
\]
and write the oracle-free units as
\[
    U_1,\ldots,U_b.
\]
Consider a unitary $W$, and define

\[
    E(W):=\frac{2I+W+W^\dagger}{4}.
\]

This is the \textit{accepting POVM element} of $W$: the $+$ outcome of controlled-Hadamard test of $W$ on a pure state $\ket{\psi}$ is obtained with
probability
\[\bra{\psi} E(W) \ket{\psi}\]
Since $R_i^f$ is a reflection (and hence self-adjoint), its accepting operator is the projector
\[
    \Pi_i^f:=E(R_i^f)=\frac{I+R_i^f}{2}.
\]
Define
\[
    P_f:=\prod_{i=1}^a\Pi_i^f,
    \qquad
    B:=\prod_{j=1}^b E(U_j).
\]
On every promised input, the units commute. It follows that $P_f$ is an orthogonal projector, $B$ is positive semidefinite with all
eigenvalues in $[0,1]$, $P_f$ commutes with $B$, and the accepting operator of the entire verifier is
\[
    A_f=P_fB.
\]
Concretely, for every normalized witness $|\psi\rangle$,
\[
\Pr[\text{verifier accepts }|\psi\rangle]
   = \langle\psi|A_f|\psi\rangle.
\]

Notice that $B$ is shared across all oracles which share the same preprocessing computation transcript $\tau$. Let $c(n)$ and $s(n)$ be the verifier's completeness and soundness, set
\[
    \theta:=\frac{c(n)+s(n)}{2},
\]
and let
\[
    Q_\theta:=\mathbf{1}_{(\theta,1]}(B)
\]
be the spectral projector of $B$ onto eigenvalues larger than $\theta$.\footnote{$Q_\theta$ is used only in the lower-bound argument; it need not be efficiently implementable.} Since $B$ is fixed on the leaf $\tau$, so is $Q_\theta$. We define
\[
    p(f):=\Tr(Q_\theta P_f).
\]
For promised $f$, the projectors $Q_\theta$ and $P_f$ commute, and hence $Q_\theta P_f$ is the projector onto the intersection of their images. Therefore, $p(f)$ is the dimension of this intersection.

Suppose first that $(g,f)$ is a YES-instance. Completeness gives
\[
    \lambda_{\max}(P_fB)\geq c(n).
\]
Because $P_f$ and $B$ commute, they have a common eigenvector on which $P_f$ has eigenvalue $1$ and $B$ has eigenvalue at least $c(n)>\theta$. This vector lies in the images of both $P_f$ and $Q_\theta$, and therefore
\[
    p(f)\geq 1.
\]

Conversely, suppose that $(g,f)$ is a NO-instance.

If $p(f)>0$, then
$\operatorname{im}P_f\cap\operatorname{im}Q_\theta$ is nonzero.
Since $P_f$ commutes with $B$, this intersection is invariant
under $B$ and contains a normalized eigenvector $|\psi\rangle$
of $B$ with eigenvalue $\lambda>\theta$.
Thus $P_f|\psi\rangle=|\psi\rangle$ and
\[
P_fB|\psi\rangle=\lambda|\psi\rangle,
\qquad \lambda>\theta>s(n),
\]
contradicting soundness. Hence $p(f)=0$.

We have therefore used the completeness--soundness gap to derive the following distinction
\[
    (g,f)\in\mathrm{Yes}
    \ \Longrightarrow\
    p(f)\geq 1,
    \qquad
    (g,f)\in\mathrm{No}
    \ \Longrightarrow\
    p(f)=0.
\]

The reflection condition turns the oracle-dependent units into projectors, while commutativity allows deriving a single projector out of them.

\paragraph{Quantum queries bound the degree.}
Suppose that $R_i^f$ makes $t_i$ quantum oracle queries. Once $g$ and the classical transcript $
\tau$ are fixed, every matrix entry of $R_i^f$ is a polynomial of degree at most $t_i$ in the truth-table values of $f$.
This is a rather standard polynomial-method observation~\cite{BealsPolynomialMethod1997}. Since $f(x)^2=1$, repeated variables can be multi-linearized without increasing the degree.

It follows that every matrix entry of $P_f$ has degree at most
\[
    T_\tau:=\sum_{i=1}^a t_i\leq T(n).
\]
The projector $Q_\theta$ is fixed on the preprocessing leaf $\tau$, so multiplying $P_f$ by $Q_\theta$ and taking the trace do not increase the degree. Hence
\[
    \deg p\leq T_\tau.
\]

\paragraph{Deriving a univariate polynomial.}

Recall that once $g$ is fixed, the Forrelation value
depends only on the Hamming distance of $f$ from $h$:
\[
    \Phi(g,f)
    =
    1-\frac{2\,\mathrm{dist}(f,h)}{N}.
\]
Following the Minsky–Papert symmetrization
method~\cite{MinskyPerceptrons,BealsPolynomialMethod1997}, we average $p$ over all functions $f$ such that $(g,f)$ has the preprocessing transcript $\tau$ and $\text{dist}(f,h)=k$.\footnote{In the proof, we reparameterize the $N-r$ unfixed
truth-table entries of $f$ using Boolean flip variables. See Section~\ref{sec:separation} for details.}

\[
    q(k):=
    \underset{\substack{f:\,\operatorname{dist}(f,h)=k\\
                       f|_F=h|_F}}{\mathbb E}
    \bigl[p(f)\bigr],
    \qquad 0\le k\le N-|F|.
\]

After fixing the coordinates in $F$, the polynomial $p$ still has
degree at most $T_\tau$. Averaging over the remaining coordinates
at Hamming distance $k$ from $h$ gives, by a standard symmetrization argument,
a univariate polynomial $q(k)$ of degree at most $T_\tau$.

For every integer $k$ in the interval
\[
    \frac{1-\beta}{2}N
    \;\leq\;
    k
    \;\leq\;
    \frac{1+\beta}{2}N,
\]
every function at distance $k$ from $h$ is a no-instance. Hence $q(k)=0$ for every integer $k$ in this range.\footnote{We additionally require $k \le N-r$.
This restriction removes at most $r$ integers from the interval.}  The interval contains
$\beta N-O(1)$ integer values of $k$, while fixing $r$ coordinates during
classical preprocessing can remove at most $r$ of them. Thus $q$ has at least
\[
    \beta N-r-O(1)
\]
distinct roots. On the other hand, $k=0$ corresponds to the original perfect yes-instance,
and therefore
\[
    q(0)=p(h)\geq 1.
\]
Thus $q$ is not the zero polynomial. Since a nonzero univariate polynomial
cannot have more roots than its degree, we obtain
\[
    T_\tau
    \geq
    \beta N-r-O(1).
\]
Using $T_\tau\leq T(n)$, $r\leq C(n)$, and $N=2^n$ gives
\[
    C(n)+T(n)
    \geq
    \beta 2^n-O(1).
\]

A standard diagonalization argument over all polynomial-time
candidate $\QIMA^{\mathcal O}$ verifiers turns the query lower bound
into an oracle separation. Specifically, there exists a classical
oracle $\mathcal O$ such that
\[
    \BQP^{\mathcal O}\not\subseteq\QIMA^{\mathcal O}.
\]

\paragraph{Robustness to negligible deviations from reflections.}
Section~\ref{sec:relax_reflection} introduces the relaxed class
$\QIMA^{O}_{\mathrm{negl\text{-}refl}}$, in which each
oracle-containing unit may be negligibly close to a
query-preserving exact companion reflection, provided that the companion reflections commute on promised inputs. It is then observed that $\QIMA^{O}_{\mathrm{negl\text{-}refl}}=\QIMA^{O}$. The reason is that replacing the units of a $\QIMA^{O}_{\mathrm{negl\text{-}refl}}$ verifier by its companion reflections yields a $\QIMA^{O}$ verifier with a noticeable completeness-soundness gap. Hence, the same query lower bound applies.

\subsubsection{Collapse to $\QMA^{\mathcal{O}}$}

\paragraph{Collapse under inverse-polynomial deviations from reflections.}
Section~4.1 shows that allowing even an inverse-polynomial
deviation from an exact reflection suffices to simulate an
arbitrary $\QMA^{\mathcal{O}}$ verifier. The main challenge is
that a $\QMA^{\mathcal{O}}$ verifier may use workspace (trusted ancilla qubits), whereas a $\QIMA^{O}$ verifier may not. Thus, in the simulation, the workspace is to be provided by the witness, and the validity of its state must be verified while preserving the commutation requirement of the verifier's units.

Let $V^{\mathcal{O}}$ be the original $\QMA^{\mathcal{O}}$ verifier, acting on a witness
register $M$ and a workspace register $A$ that would ordinarily
be initialized to $|0^a\rangle$. Define
\[
    P := I_M \otimes |0^a\rangle\langle 0^a|_A,
    \qquad
    Q := (V^{\mathcal{O}})^\dagger
         \Pi_{\mathrm{acc}} V^{\mathcal{O}},
\]
where $\Pi_{\mathrm{acc}}$ projects onto the accepting output
subspace. Thus, $P$ projects onto inputs with correctly initialized
workspace, while $Q$ projects onto inputs that
$V^{\mathcal{O}}$ maps into the accepting output subspace. Testing $Q$ alone would allow Merlin to exploit
improperly initialized workspace, but testing $P$ separately
need not commute with testing $Q$.

We instead combine these two projectors into a single unit,
using a variation of the two-projector construction
underlying Marriott--Watrous QMA amplification~\cite{Marriott2005}. Set
\[
    R_P := 2P-I,
    \qquad
    W_\theta := R_P e^{-i\theta Q}e^{i\theta P}.\footnote{Here, consider the spectral decomposition of $P=SDS^\dag$, where $D$ is a diagonal matrix. Then $e^{i\theta P}=Se^{i\theta D}S^\dag$, where in $e^{i\theta D}$ each value $\lambda$ along the diagonal of $D$ becomes $e^{i\theta \lambda}$.}
    \]
For a correctly initialized input $|\psi,0^a\rangle$, we have
\[
    \|Q|\psi,0^a\rangle\|^2
    =
    \Pr[V^{\mathcal{O}}\text{ accepts }|\psi\rangle].
\]
Thus, if this acceptance probability is close to one, the input
lies almost entirely in the image of $Q$. On such an input,
$e^{i\theta P}$ applies the phase $e^{i\theta}$, while
$e^{-i\theta Q}$ applies the opposite phase to almost all of
the state. The two shifts therefore nearly cancel. This
connects the action of the phase shifts to the original
verifier's acceptance probability. The factor $R_P$
is essential for soundness. Without it, every state in
$\ker P\cap\ker Q$ would be fixed by the two phase shifts
and would pass the Hadamard test with certainty.

A two-dimensional block analysis shows that the maximum
acceptance probability of the Hadamard test for $W_\theta$
is an increasing function of the original verifier's optimal
acceptance probability. In particular, an original
completeness--soundness gap between $2/3$ and $1/3$ becomes
a gap of $\Omega(\theta^2)$. This analysis maximizes over all
states of $M\otimes A$, so it already accounts for arbitrary
workspace supplied by Merlin. At the same time,
\[
    \|W_\theta-R_P\|_{\mathrm{op}}\leq 2\theta.
\]
Choosing $\theta$ sufficiently small therefore meets any
prescribed inverse-polynomial permitted distance from a reflection, while retaining an
inverse-polynomial acceptance gap. The construction is
query-efficient because
\[
    e^{-i\theta Q}
    = (V^{\mathcal{O}})^\dagger
      e^{-i\theta\Pi_{\mathrm{acc}}}V^{\mathcal{O}},
\]
and the remaining operations are oracle-free phases conditioned
on the workspace being zero. Since the verifier tests only
this one unit, commutativity is automatic.

\paragraph{Collapse with a single oracle query.}
The verifier unit $W_\theta$ used in Section~\ref{subsec:inverse-poly-reflection} for simulating a $\QMA^{\mathcal{O}}$ verifier, may contain polynomially many oracle queries.
Section~\ref{subsec:one-query-collapse} reduces this to one query when the reflection
requirement is removed entirely. We start with the same unit
$W_\theta$, now taking $\theta$ to be a fixed constant, and
apply a cyclic version of the familiar clock construction
underlying the Feynman--Kitaev circuit-to-Hamiltonian
reduction~\cite{Kitaev2002}. This cyclic version also
appears in~\cite[Section~6]{JW07}.

We arrange the unit into $L$
steps of the form $G_t\mathcal{O}^{(R)}$, where each $G_t$
is oracle free and every query uses the same register $R$.
Here $L$ is polynomially bounded. Add a "clock register"\footnote{A clock register is a quantum
register whose basis states $|t\rangle$ label the steps of a
computation. It allows the propagation unitary to select the
operation for step $t$ and advance the label to the next step,
coherently even when the clock is in a superposition of labels.} $C$
and define a unitary that performs one step and advances
the clock cyclically:
\[
    S^{\mathcal{O}}
    \bigl(|t\rangle_C\otimes|\psi\rangle\bigr)
    =
    |t+1 \bmod L\rangle_C
    \otimes G_t\mathcal{O}^{(R)}|\psi\rangle.
\]
The oracle operation is the same for every clock value.
Consequently, implementing $S^{\mathcal{O}}$ requires only
one oracle query, followed by oracle-free operations
controlled by the clock.

The clock preserves the information encoded in the spectrum
of $W_\theta$: the eigenvalues of $S^{\mathcal{O}}$ are exactly
the $L$th roots of the eigenvalues of $W_\theta$. Since a
Hadamard test accepts an eigenvector with eigenvalue
$e^{i\varphi}$ with probability $(1+\cos\varphi)/2$, taking
these roots induces an increasing transformation of the
optimal acceptance probability. The completeness--soundness
gap decreases by at most a factor of $L^2$, and hence remains
inverse polynomial.

The spectral analysis covers the entire enlarged register,
so both the workspace and the clock may be supplied by the witness.
This gives a simulation of a $\QMA^{\mathcal{O}}$ verifier using a single unit and a single oracle
query, with no trusted workspace beyond the Hadamard-test
control. The clock
construction does not retain the near-reflection guarantee
of Section~\ref{subsec:inverse-poly-reflection}.

\paragraph{Collapse with one trusted active qubit.}
Finally, Section~4.3 shows that one trusted ancilla
qubit (in addition to the control qubits used in the Hadamard test) suffices to recover $\QMA$ in the plain model and
$\QMA^{\mathcal{O}}$ in the oracle model, with an inverse-polynomial completeness-soundness gap. 

For the plain model, we rely on the result of Nagaj, Hangleiter,
Eisert, and Schwarz~\cite{Nagaj_2021} that the pinned commuting
local-Hamiltonian problem is $\QMA$-complete. In this problem,
the Hamiltonian terms commute, but the energy is minimized
only over states in which one designated qubit is fixed to
$|0\rangle$.

Their reduction starts from a $\QMA$-hard family of Hamiltonians
of the form
\[
    H = \sum_i A_i + \sum_j B_j,
\]
where the $A_i$ commute with one another and the $B_j$ commute
with one another, but terms from different families need not
commute. They add one qubit and define
\[
    \widetilde{H}
    =
    \sum_i A_i \otimes |+\rangle\langle+|
    +
    \sum_j B_j \otimes |-\rangle\langle-|.
\]
All terms of $\widetilde{H}$ commute: terms within each family
commute by assumption, and products of terms from different
families vanish because
$|+\rangle\langle+|\,|-\rangle\langle-|=0$.
Moreover, for every state $|\psi\rangle$,
\[
    \langle\psi,0|\widetilde{H}|\psi,0\rangle
    =
    \frac{1}{2}\langle\psi|H|\psi\rangle.
\]
Thus, the minimum energy of $\widetilde{H}$ with the additional
qubit pinned to $|0\rangle$ is exactly half the ground-state
energy of $H$. 

We translate this pinned Hamiltonian test into
commuting reflection tests using one trusted active qubit. We turn each local energy test into a reflection test using
an auxiliary qubit: when this qubit is initialized to
$|0\rangle$, the new test reproduces the original acceptance
probability. Since the witness supplies these auxiliary qubits,
we also check their initialization and the Hamiltonian's
pinning condition. We make the two groups of tests commute
using the same construction as above: the energy tests act
only when the trusted qubit is in $|+\rangle$, and the
initialization checks act only when it is in $|-\rangle$.
Preparing the trusted qubit in
$|0\rangle=(|+\rangle+|-\rangle)/\sqrt{2}$ makes the overall
acceptance probability the average of the probabilities
of passing the two groups of tests.

In the oracle model, a simpler construction suffices: the trusted
qubit records whether the workspace supplied by Merlin is
correctly initialized. We combine this check with the original
verifier's acceptance condition into a single exact reflection,
preserving completeness and soundness.

\subsection{AI Usage Disclosure}
We came up with the idea of proving an oracle separation between $\QIMA$ and $\QMA$, and identified Spectral Forrelation (see \cite{bostanci2026separatingquantumclassicaladvice}) and Forrelation as leading candidate problems. We then proved a separation for limited relativized analogues of $\QIMA$, and identified the reflection model as a conceptually appealing setting that could give way to a separation. ChatGPT 5.6-sol was then used for completing the separation proof. We came up with the approach for proving sections \ref{sec:relax_reflection}, \ref{subsec:inverse-poly-reflection} and \ref{subsec:ancilla-collapse}, and used ChatGPT 5.6-sol for formalizing the proofs. ChatGPT 6 Astra was then used for proving sections \ref{subsec:restricted-address-support-lower-bound} and \ref{subsec:one-query-collapse}, and for improving the writing throughout the paper. The exposition and interpretation of the results were developed by us, and we take full responsibility for them and for the correctness of the results themselves. 

\subsection{Organization}

Section 2 introduces the necessary preliminaries. Section 3 presents the proof that Forrelation $\notin\QIMA^{\cal{O}}$. Section 4 contains results which demonstrate that natural extensions of our $\QIMA^{\cal{O}}$ model collapse it to $\QMA^{\cal{O}}$, thus showing a tightness property of our model.

\section{Preliminaries}
\label{sec:preliminaries}
\subsection{The Class $\QIMA_k$}

\begin{definition}[$\QIMA_k$]\label{def:QIMA_k}
A promise problem $L=(L_{\mathrm{yes}},L_{\mathrm{no}})$ is in
$\QIMA_k$ if there exists a polynomial-time uniform family of
quantum verifiers with the following properties.

\begin{itemize}
    \item \textbf{Witness.}
    On an input $x\in\{0,1\}^n$, the verifier receives a
    $\poly(n)$-qubit quantum witness $\ket{\psi}$.

    \item \textbf{Commutation units.}
    Given an input $x\in\{0,1\}^n$, the verifier computes in classical-polynomial time $m=\poly(n)$ reflections
    \[
    R_1,\ldots,R_m,
    \]
    acting on the witness register, where each $R_i$ is $k$-local and
    the reflections mutually commute:
    \[
    [R_i,R_j]=0
    \qquad
    \text{for all }i,j.
    \]

    \item \textbf{Verification.}
    For each $R_i$, the verifier introduces a fresh control qubit
    initialized to $\ket{+}$, applies controlled-$R_i$, and measures the
    control qubit in the $X$ basis. The verifier accepts if and only if
    every measurement outcome is $+$.

    \item \textbf{Completeness and soundness.}
    
    \begin{align*}
        x\in L_{\mathrm{yes}}
        &\Longrightarrow
        \exists\,\ket{\psi}\quad
        \Pr[V_x(\ket{\psi})\text{ accepts}]\,=\, 1,\\
        x\in L_{\mathrm{no}}
        &\Longrightarrow
        \forall\,\ket{\psi}\quad
        \Pr[V_x(\ket{\psi})\text{ accepts}]\,=\, 0.
    \end{align*}
\end{itemize}
\end{definition}

\begin{remark}
Recall that we define $\QIMA=\bigcup_{k\in\mathbb{N}} \QIMA_k$. 
\end{remark}

\begin{remark} \label{remark: gate-convention}
We allow constant-local unitaries specified by exact algebraic
matrices as elementary gates, with descriptions generated uniformly
in polynomial time. Local Hamiltonian terms use the same exact
representation.
\end{remark}

\subsection{Alternating Projectors}
\label{subsec:mw-two-projector}

Marriott and Watrous designed algorithms which rely on alternating projectors for
witness-preserving QMA amplification~\cite{Marriott2005}. We discuss their approach below using the formulation via Jordan’s lemma and products of reflections from the work of Nagaj et al.\cite{NagajFastAmplification2009}.

Let $V$ be a unitary QMA verifier acting on a witness register $M$
and an $a$-qubit workspace register $A$, where $A$ is initialized to
$\lvert 0^a\rangle$. Let $\Pi_{\mathrm{acc}}$ be the projector onto the
accepting output subspace, and define
\begin{equation}
    P
    :=
    I_M \otimes \lvert 0^a\rangle\!\langle 0^a\rvert_A,
    \qquad
    Q
    :=
    V^\dagger \Pi_{\mathrm{acc}} V.
\end{equation}
Thus, $P$ projects onto properly initialized inputs, while $Q$ projects
onto inputs that $V$ maps to the accepting subspace. For every witness
$\lvert\psi\rangle\in M$,
\begin{equation}
    \Pr\!\left[V \text{ accepts } \lvert\psi\rangle\right]
    =
    \langle \psi,0^a \rvert Q \lvert \psi,0^a\rangle.
\end{equation}
Equivalently, the usual QMA acceptance operator on $M$ is
\begin{equation}
    \Omega
    :=
    \bigl(I_M\otimes \langle 0^a\rvert\bigr)
    Q
    \bigl(I_M\otimes \lvert 0^a\rangle\bigr),
\end{equation}
and the maximal eigenvalue of $\Omega$ can be described as follows 
\begin{equation}
    \lambda_{\max}(\Omega)
    =
    \lambda_{\max}
    \left(
        PQP
    \right).
\end{equation}

\noindent The interaction between $P$ and $Q$ is described by Jordan's lemma.

\begin{lemma}[Jordan's lemma for two projectors]
\label{lem:jordan-two-projectors}
Let $P$ and $Q$ be orthogonal projectors on a finite-dimensional
Hilbert space $\mathcal H$. There is an orthogonal decomposition
\begin{equation}
    \mathcal H
    =
    \bigoplus_{\ell} \mathcal H_\ell
\end{equation}
into subspaces of dimension at most two, each invariant under both
$P$ and $Q$.

On every two-dimensional block, there is an orthonormal basis in which
\begin{equation}
    P
    =
    \begin{pmatrix}
        1 & 0\\
        0 & 0
    \end{pmatrix},
    \qquad
    Q
    =
    \begin{pmatrix}
        \lambda & \sqrt{\lambda(1-\lambda)}\\
        \sqrt{\lambda(1-\lambda)} & 1-\lambda
    \end{pmatrix}
\end{equation}
for some $\lambda\in(0,1)$. On a one-dimensional block, each of
$P$ and $Q$ acts as either $0$ or $1$.
\end{lemma}

In the QMA interpretation, $\lambda$ is the acceptance probability for the unit vector \(|p\rangle=|\psi\rangle_M|0^a\rangle_A\) spanning \(\operatorname{im}P\) within that block. Thus, \(\lambda=\langle p|Q|p\rangle\). The amplification algorithms by Marriott--Watrous rely on alternating application of $P$ and $Q$, and can be analyzed using the Jordan decomposition.\footnote{For a detailed account of the amplification algorithms, see \cite{Marriott2005} or \cite{NagajFastAmplification2009}.}

The same geometry can be
described coherently using the reflections
\begin{equation}
    R_P := 2P-I,
    \qquad
    R_Q := 2Q-I.
\end{equation}
Writing $\lambda=\cos^2\varphi$, the product $R_QR_P$ has eigenvalues
$e^{\pm 2i\varphi}$ on the two-dimensional block which corresponds to $\lambda$. Thus,
the original QMA acceptance probability is encoded in the rotation
generated by the two reflections.

Equivalently, the reflections arise from full selective phase shifts:
\begin{equation}
    e^{i\pi P}=-R_P,
    \qquad
    e^{-i\pi Q}=-R_Q,
\end{equation}
and hence
\begin{equation}
    e^{-i\pi Q}e^{i\pi P}
    =
    R_QR_P.
\end{equation}
In Section~\ref{subsec:inverse-poly-reflection}, we use the same
two-projector mechanism, but replace these full phase shifts by small
phase shifts.

\section{Oracle Separation}\label{sec:separation}

In Appendix~\ref{sec:forrelation_in_bqp}, we present the simple $\BQP^{\cal{O}}$ algorithm for Forrelation that was presented by Aaronson and Ambainis~\cite{aaronson2014forrelationproblemoptimallyseparates}. We can thus state the following proposition:

\begin{proposition}[Forrelation$\,\in \BQP^{\cal{O}}$]
The Forrelation promise problem belongs to \(\BQP^{O}\).
\end{proposition}

Next, we provide the complete argument for the query lower bound for a $\QIMA^{\cal{O}}$-verifier for Forrelation.

\subsection{Perfect Forrelation Pairs} \label{subsec:perfect_forr_pairs}

The lower-bound argument uses a perfect Forrelation pair. Girish and Servedio point out that a Boolean function $f$ has a companion Boolean function $g$ such that $\Phi(f,g)=1$, if and only if $f$ is a bent function~\cite{girish2025forrelationextremallyhard} (section 1.4.2). For completeness, we give an explicit construction using a
simple instance of the classical Maiorana--McFarland family
of bent functions~\cite{Dillon1974}, also used by Girish and
Servedio~\cite{girish2025forrelationextremallyhard}. Bent functions were defined by Rothaus~\cite{rothaus1976}.

\begin{definition}[Bent functions] 
For a sign function $u:\{0,1\}^n\to\{-1,+1\}$, define its normalized
Walsh--Hadamard transform by
\[
    \widehat u(z):=(H_Nu)(z)
      =\frac{1}{\sqrt N}
        \sum_{x\in\{0,1\}^n}(-1)^{x\cdot z}u(x).
\]
The function $u$ is called \emph{bent} if
\[
    |\widehat u(z)|=1
    \qquad\text{for every }z\in\{0,1\}^n.
\]
\end{definition}
Equivalently, every unnormalized Walsh coefficient of $u$ has absolute value
$\sqrt N$.  In particular, if $u$ is bent, then every entry of $H_Nu$ is either
$-1$ or $+1$, so $H_Nu$ is itself the truth table of a Boolean sign function.

This gives a direct correspondence between bent functions and perfect
Forrelation pairs~\cite{girish2025forrelationextremallyhard}.  If $u$ is bent and $v:=H_Nu$, then
$v\in\{-1,+1\}^N$ and
\[
    \Phi(v,u)
      =\frac{1}{N}v^{\mathsf T}H_Nu
      =\frac{1}{N}v^{\mathsf T}v
      =1.
\]
Conversely, suppose that $u,v\in\{-1,+1\}^N$ satisfy $\Phi(v,u)=1$.  Since
$H_N$ is orthogonal,
\[
    v^{\mathsf T}H_Nu
      \leq \|v\|_2\,\|H_Nu\|_2
      =N.
\]
Equality holds, and both vectors have Euclidean norm $\sqrt N$; equality in
Cauchy--Schwarz therefore implies $v=H_Nu$.  Hence $H_Nu$ is Boolean and $u$ is
bent.

For completeness, we use a simple instance of the
Maiorana--McFarland family of bent functions~\cite{Dillon1974}. When $n$ is even one may take the explicit quadratic bent
function
\[
    g^{\star}(x)
      :=(-1)^{x_1x_2+x_3x_4+\cdots+x_{n-1}x_n}.
\]
Indeed, for every $a,b\in\{0,1\}$,
\[
    \frac{1}{2}
    \sum_{r,s\in\{0,1\}}
      (-1)^{rs+ar+bs}
      =(-1)^{ab}.
\]
The Walsh--Hadamard transform of $g^{\star}$ therefore factors over the
$n/2$ disjoint pairs of variables, giving
\[
\begin{aligned}
    (H_Ng^{\star})(z)
      &=\prod_{j=1}^{n/2}
        \left(
        \frac{1}{2}
        \sum_{r,s\in\{0,1\}}
        (-1)^{rs+z_{2j-1}r+z_{2j}s}
        \right) \\
      &=\prod_{j=1}^{n/2}(-1)^{z_{2j-1}z_{2j}}
       =g^{\star}(z).
\end{aligned}
\]
Thus $g^{\star}$ is self-dual and
$\Phi(g^{\star},g^{\star})=1$.

\begin{remark}
Perfect pairs of this form can exist only when \(n\) is even. Indeed, if
\[
H_Ng\in\{-1,+1\}^N,
\]
then every unnormalized Walsh coefficient of \(g\) equals
\[
\pm\sqrt{N}
=
\pm 2^{n/2}.
\]
An unnormalized Walsh coefficient is an integer, whereas \(2^{n/2}\) is
not an integer when \(n\) is odd.

This causes no difficulty for the oracle separation. Proving the lower
bound for every even input length already rules out a polynomial-query
verifier.
\end{remark}

Next, we prove that for a perfect Forrelation pair $(g,h)$, flipping bits in the truth table of one of the functions causes the Forrelation value to change by an additive factor that is linear in the number of bits flips. This lemma will be used in our oracle separation proof.

\begin{lemma}[Forrelation under Hamming perturbations] \label{lemma:forrelation-under-hamming}
Suppose that
\[
g,h:\{0,1\}^n\to\{-1,+1\}
\]
satisfy
\[
H_Ng=h.
\]
Let \(f:\{0,1\}^n\to\{-1,+1\}\) differ from \(h\) on exactly \(k\)
inputs. Then
\[
\Phi(f,g)=1-\frac{2k}{N}.
\]
\end{lemma}

\begin{proof}
Since \(f(x),h(x)\in\{-1,+1\}\),
\[
f(x)h(x)
=
\begin{cases}
1,&\text{if }f(x)=h(x),\\
-1,&\text{if }f(x)\neq h(x).
\end{cases}
\]
There are \(N-k\) coordinates on which \(f\) and \(h\) agree, and \(k\)
coordinates on which they disagree. Hence
\[
\begin{aligned}
f^{\mathsf T}h
&=
\sum_{x\in\{0,1\}^n}f(x)h(x)\\
&=
(N-k)-k\\
&=
N-2k.
\end{aligned}
\]
Using \(H_Ng=h\), we obtain
\[
\begin{aligned}
\Phi(f,g)
&=
\frac{1}{N}f^{\mathsf T}H_Ng\\
&=
\frac{1}{N}f^{\mathsf T}h\\
&=
1-\frac{2k}{N}.
\end{aligned}
\]
\end{proof}

\subsection{Forrelation Query Lower-Bound}

\begin{theorem} \label{thm:qima_lower_bound}
Consider the Forrelation promise problem from Definition~\ref{def:forrelation}, with parameters $(\alpha,\beta)$.
Suppose a \(\QIMA^{\mathcal{O}}\) verifier $V^{\mathcal{O}}$ decides this problem with the soundness and completeness parameters
\[
c(n)>s(n).
\]
Then for even values of $n$,
\[
C(n)+T(n)\ge \beta N-O(1).
\]
In particular, no such verifier can have both classical and quantum polynomial query complexity.
\end{theorem}

\begin{proof}
A verifier's oracle access is as in Definition~\ref{def:oracle-access}. Assume first that \(n\) is even. Choose a bent function
\[
g_\star:\{0,1\}^n\to\{-1,+1\}
\]
such that
\[
h:=H_Ng_\star\in\{-1,+1\}^N.
\]
Then
\[
\Phi(h,g_\star)
= \frac1N h^{\mathsf T}H_Ng_\star
= \frac1N h^{\mathsf T}h
=1,
\]
so \((h,g_\star)\) is a \textsc{Yes} instance.

Run the classical preprocessing phase of $V^{\mathcal{O}}$ on \((h,g_\star)\), and let \(\tau\) be the resulting adaptive query transcript. Let
\[
F\subseteq\{0,1\}^n
\]
be the set of positions of \(f\) queried along this transcript, and write $
r:=|F|\le C(n)$. We keep \(g=g_\star\) fixed and consider any \(f\) satisfying
\begin{equation}
f(x)=h(x)
\qquad
\text{for every }x\in F.
\label{eq:fixed-subcube}
\end{equation}
Every such \(f\) produces the same preprocessing transcript \(\tau\): inductively, every classical query receives the same answer as on \((h,g_\star)\), and therefore the next query is also the same. Consequently, throughout this subcube the preprocessing algorithm outputs exactly the same quantum circuit.

Partition this circuit into oracle-querying reflections
\[
R_1^f,\ldots,R_a^f
\]
and oracle-free unitaries
\[
U_1,\ldots,U_b.
\]

For a unitary \(W\), define the operator

\[E_W:=\frac{2I+W+W^\dagger}{4}.\]

Then, the \(+\) outcome of its controlled-Hadamard test on a pure state $\ket{\psi}$ is obtained with probability

\[\bra{\psi}E_W\ket{\psi}.\]

For a reflection $R_i^f$, this operator reduces to the projector
\[
\Pi_i^f
:=
E_{R_i^f}
=
\frac{I+R_i^f}{2}.
\]
Define
\begin{equation}
P_f
:=
\prod_{i=1}^a \Pi_i^f
\label{eq:reflection-projector}
\end{equation}
and
\begin{equation}
B
:= \prod_{j=1}^b E_{U_j}= 
\prod_{j=1}^b
\frac{2I+U_j+U_j^\dagger}{4}.
\label{eq:oracle-free-operator}
\end{equation}

On every promised input, all completed units are commuting unitaries. They are therefore simultaneously diagonalizable. In particular, all the factors in \eqref{eq:reflection-projector} and \eqref{eq:oracle-free-operator} commute, and the total accepting POVM element is
\begin{equation}
A_f=P_fB.
\label{eq:total-acceptance}
\end{equation}

That is, for a witness state $\ket{\psi}$, its probability of acceptance by the verifier is 

\[\bra{\psi}A_f\ket{\psi}\]

Moreover, \(P_f\) is an orthogonal projector and
\[
0\le B\le I.
\]

Choose a threshold
\[
\theta:=\frac{c(n)+s(n)}{2},
\]
and let
\[
Q
:=
\mathbf 1_{(\theta,1]}(B)
\]
be the spectral projector of \(B\) onto eigenvalues larger than \(\theta\). The operator \(Q\) need not be efficiently implementable; it is used only in the lower-bound argument. Since \(B\) is fixed throughout the transcript \(\tau\), so is \(Q\).

For every promised \(f\) satisfying \eqref{eq:fixed-subcube}, the projector \(P_f\) commutes with \(B\), and hence also with \(Q\). Indeed, \(Q=\mathbf 1_{(\theta,1]}(B)\) is a spectral projector of \(B\). Since \(P_f\) commutes with \(B\), it preserves every eigenspace of \(B\). It therefore preserves the direct sum of the eigenspaces of \(B\) with eigenvalues in \((\theta,1]\), as well as its orthogonal complement. Hence \(P_f\) commutes with the projector onto this subspace, namely \(Q\). Thus, $QP_f$ is a projector.

Define
\[
p(f)
:=
\operatorname{Tr}(QP_f).
\]

We claim that
\begin{equation}
\begin{aligned}
(f,g_\star)\in\textsc{Yes}
&\Longrightarrow p(f)\ge1,\\
(f,g_\star)\in\textsc{No}
&\Longrightarrow p(f)=0.
\end{aligned}
\label{eq:yes-no-separation}
\end{equation}

Indeed, suppose first that \((f,g_\star)\) is a \textsc{Yes} instance. By completeness and \eqref{eq:total-acceptance},
\begin{equation}\label{eq:completness}
\lambda_{\max}(P_fB)\ge c(n).
\end{equation}
Since $P_f$ is a projector, its only eigenvalues can be $0,1$. Additionally, \(P_f\) and \(B\) commute, and due to \eqref{eq:completness} they have a common eigenvector \(\ket{\psi}\) satisfying
\[
P_f\ket{\psi}=\ket{\psi},
\qquad
B\ket{\psi}=\lambda\ket{\psi}
\]
for some
\[
\lambda\ge c(n)>\theta.
\]
Therefore \(Q\ket{\psi}=\ket{\psi}\), so \(QP_f\neq0\). Thus, \(QP_f\) is a non-zero projector, and hence
\[
p(f)=\operatorname{Tr}(QP_f)\ge1.
\]

Now suppose that \((f,g_\star)\) is a \textsc{No} instance. Soundness gives
\[
\lambda_{\max}(P_fB)\le s(n).
\]

If $QP_f\neq0$, then
$\operatorname{im}P_f\cap\operatorname{im}Q$ is nonzero.
Since $P_f$ commutes with $B$, this intersection is invariant
under $B$ and contains a normalized eigenvector $|\psi\rangle$
of $B$ with eigenvalue $\lambda>\theta$.
Thus $P_f|\psi\rangle=|\psi\rangle$ and
\[
P_fB|\psi\rangle=\lambda|\psi\rangle,
\qquad \lambda>\theta>s(n),
\]
contradicting soundness. Thus $QP_f=0$, proving \eqref{eq:yes-no-separation}.

We next bound the degree of \(p(f)\) as a polynomial in the oracle values. Fix the preprocessing transcript \(\tau\). The quantum circuit output by the preprocessing algorithm may depend on the classical oracle answers appearing in \(\tau\). However, every \(f\) satisfying \eqref{eq:fixed-subcube} produces exactly the same transcript \(\tau\). Therefore, throughout the subcube defined by \eqref{eq:fixed-subcube}, the following are fixed:

\begin{itemize}
\item the sequence of gates in every generated circuit;
\item the subsets of witness qubits on which they act;
\item the query registers of the quantum oracle calls.
\end{itemize}

Only the actions of the quantum oracle calls continue to depend on \(f\) even under the fixed $\tau$. Suppose that the reflection \(R_i^f\) makes \(t_i\) quantum oracle queries. On the fixed transcript \(\tau\), write its circuit as
\begin{equation}
R_i^f
=
V_{i,t_i}^{\tau}
O_{f,g_\star}^{(Q_{i,t_i})}
V_{i,t_i-1}^{\tau}
\cdots
O_{f,g_\star}^{(Q_{i,1})}
V_{i,0}^{\tau}.
\label{eq:fixed-circuit-decomposition}
\end{equation}
Here:

\begin{itemize}
\item each \(V_{i,j}^{\tau}\) is a fixed unitary determined by the transcript \(\tau\);
\item \(V_{i,j}^{\tau}\) may act on an arbitrary subset of the witness qubits;
\item \(O_{f,g_\star}^{(Q_{i,j})}\) denotes a phase-oracle call on an arbitrary designated query subsystem \(Q_{i,j}\);
\item every operation is understood to act as the identity on all untouched qubits.
\end{itemize}

Fix the computational basis of the entire witness register, including all registers on which the query operations act. A phase-oracle call is diagonal in this basis. Thus, for every basis state \(\ket{y}\),
\begin{equation}
\label{eq:oracle_on_comp_states}
O_{f,g_\star}^{(Q_{i,j})}\ket{y}
=
\chi_{i,j}(y)\ket{y},
\end{equation}

Here $\chi_{i,j}(y)$ is either $f(x)$ or $g_\star(x)$,
where $x$ is the address specified by the query register in
$|y\rangle$, or $1$ if the query is controlled and its
controls are not all in state $|1\rangle$.
Since $g_\star$ is fixed, each such factor is either a
variable $f(x)$ or a constant.

We next apply the polynomial method for quantum query complexity
\cite{BealsPolynomialMethod1997}.  The standard observation underlying this method
is that the amplitudes of a $t$-query quantum computation are polynomials
of degree at most $t$ in the oracle values.  We use the same observation
entrywise for each oracle-querying reflection. We observe that

\[
\begin{aligned}
\bra{u}R_i^f\ket{v}
=
\sum_{y_0,\ldots,y_{t_i-1}}
&\bra{u}V_{i,t_i}^{\tau}
O_{f,g_\star}^{(Q_{i,t_i})}
\ket{y_{t_i-1}}
\\
&\quad\cdot
\prod_{j=1}^{t_i-1}
\bra{y_j}V_{i,j}^{\tau}
O_{f,g_\star}^{(Q_{i,j})}
\ket{y_{j-1}}
\cdot
\bra{y_0}V_{i,0}^{\tau}\ket{v}.
\end{aligned}
\]

By applying \eqref{eq:oracle_on_comp_states}, and setting \(y_{-1}:=v\) and \(y_{t_i}:=u\) for notational convenience, the preceding
expression becomes
\[
\bra{u}R_i^f\ket{v}
=
\sum_{y_0,\ldots,y_{t_i-1}}
\left(
\prod_{j=0}^{t_i}
\bra{y_j}V_{i,j}^{\tau}\ket{y_{j-1}}
\right)
\left(
\prod_{j=1}^{t_i}
\chi_{i,j}(y_{j-1})
\right).
\]

For each summand, its first product consists entirely of fixed matrix entries
of the unitaries \(V_{i,j}^{\tau}\), while its second product contains
exactly one oracle value for each of the \(t_i\) oracle calls.
 All matrix entries of the \(V_{i,j}^{\tau}\)'s are constants once \(\tau\) is fixed. The factors involving \(g_\star\) are also constants and may be absorbed into the coefficients. Consequently, \(\bra{u}R_i^f\ket{v}\) is a polynomial in the values
\[
\{f(x):x\in\{0,1\}^n\}
\]
of total degree at most \(t_i\).

A single oracle position may be queried more than once, producing powers such as \(f(x)^r\). Since
\[
f(x)\in\{-1,+1\},
\qquad
f(x)^2=1,
\]
the polynomial may be multilinearized without increasing its degree.

Let
\[
T_\tau:=\sum_{i=1}^a t_i\le T(n).
\]
Every matrix entry of
\[
P_f
=
\prod_{i=1}^a\frac{I+R_i^f}{2}
\]
is therefore a polynomial of degree at most \(T_\tau\). Since \(Q\) is fixed throughout the preprocessing leaf \(\tau\), its matrix entries are constants independent of \(f\). Multiplication by \(Q\) and taking the trace do not increase the degree. Hence
\[
p(f)
=
\operatorname{Tr}(QP_f)
\]
is a multilinear polynomial in the values of \(f\) satisfying
\[
\deg p\le T_\tau.
\]

Let
\[
\overline F:=\{0,1\}^n\setminus F,
\qquad
M:=|\overline F|=N-r.
\]
For each \(x\in\overline F\), introduce a bit \(b_x\in\{0,1\}\), and define
\begin{equation}
f_b(x)
=
\begin{cases}
h(x),&x\in F,\\
h(x)(1-2b_x),&x\in\overline F.
\end{cases}
\label{eq:subcube-parameterization}
\end{equation}
Every \(f_b\) agrees with \(h\) on all classical queries made in \(\tau\), and hence remains in the same preprocessing leaf.

Recall that the oracle values \(f(x)\) are sign variables, satisfying
\(f(x)^2=1\). We reparameterize the unfixed oracle values by Boolean
variables \(b_x\in\{0,1\}\).
This is an affine substitution, so substituting \(f=f_b\) into \(p(f)\)
does not increase its total degree. The resulting polynomial is defined
on the Boolean cube \(\{0,1\}^{M}\), where \(b_x^2=b_x\). It may therefore
be multilinearized, without changing its values on the Boolean cube or
increasing its degree.

Write the multilinearized polynomial as

\[
\widetilde p(b)
=
\sum_{\substack{S\subseteq\overline F\\ |S|\le T_\tau}}
c_S\prod_{x\in S}b_x.
\]

We now symmetrize $\widetilde p$ over Hamming layers, following the
Minsky--Papert symmetrization method
\cite{MinskyPerceptrons,BealsPolynomialMethod1997}. For every \(k\in\{0,\ldots,M\}\), define
\[
q(k)
:=
\mathbb E_{\substack{b\in\{0,1\}^{M}\\ |b|=k}}
\widetilde p(b).
\]

If \(S\subseteq\overline F\) has size \(d\), then
\[
\mathbb E_{|b|=k}
\left[
\prod_{x\in S}b_x
\right]
=
\frac{\binom{M-d}{k-d}}{\binom Mk}
=
\frac{\binom kd}{\binom Md}.
\]

Indeed, the product \(\prod_{x\in S} b_x\) equals \(1\) exactly when
\(b_x=1\) for every \(x\in S\). Among the \(\binom{M}{k}\) binary strings
of weight \(k\), precisely \(\binom{M-d}{k-d}\) have this property: the
\(d=|S|\) positions in \(S\) are fixed to \(1\), and the remaining
\(k-d\) ones may be chosen among the other \(M-d\) positions. 

Therefore
\begin{equation}
q(k)
:= \mathbb{E}_{|b|=k} \left[
\widetilde{p}(b)
\right]=
\sum_{\substack{S\subseteq\overline F\\ |S|\le T_\tau}}
c_S
\frac{\binom{k}{|S|}}{\binom{M}{|S|}}.
\label{eq:symmetrized-polynomial}
\end{equation}
For fixed \(d\),
\[
\binom{k}{d}
=
\frac{k(k-1)\cdots(k-d+1)}{d!}
\]
is an ordinary polynomial in \(k\) of degree \(d\). Thus \eqref{eq:symmetrized-polynomial} defines a univariate polynomial over \(\mathbb C\) satisfying
\begin{equation}
\deg q\le T_\tau.
\label{eq:symmetrized-degree}
\end{equation}

If \(|b|=k\), then \(f_b\) differs from \(h\) on exactly \(k\) positions. Since \(h=H_Ng_\star\),
\[
\Phi(f_b,g_\star)
=
\frac1N f_b^{\mathsf T}h
=
1-\frac{2k}{N}.
\]

Consequently, every integer \(k\) satisfying
\begin{equation}
\frac{1-\beta}{2}N
\le k\le
\frac{1+\beta}{2}N
\label{eq:no-instance-interval}
\end{equation}
and
\[
k\le M
\]
corresponds entirely to \textsc{No} instances. By \eqref{eq:yes-no-separation},
\[
q(k)=0
\]
for every such integer.

The interval in \eqref{eq:no-instance-interval} contains \(\beta N-O(1)\) integers. Intersecting it with \(\{0,\ldots,M\}\) removes at most
\[
N-M=r
\]
integers. Hence \(q\) has at least
\begin{equation}
\beta N-r-O(1)
\label{eq:root-count}
\end{equation}
distinct roots.

At \(k=0\), we recover the canonical \textsc{Yes} instance \((h,g_\star)\). By \eqref{eq:yes-no-separation},
\begin{equation}
q(0)=p(h)\ge1.
\label{eq:nonzero-polynomial}
\end{equation}
Thus \(q\) is not the zero polynomial.

A nonzero univariate polynomial over \(\mathbb C\) has at most as many distinct roots as its degree. Combining \eqref{eq:symmetrized-degree}, \eqref{eq:root-count}, and \eqref{eq:nonzero-polynomial}, we obtain
\[
T_\tau
\ge
\beta N-r-O(1).
\]
Since
\[
T_\tau\le T(n)
\qquad\text{and}\qquad
r\le C(n),
\]
it follows that
\[
C(n)+T(n)\ge\beta N-O(1).
\]

Since \(N=2^n\), the verifier cannot have both polynomial classical-query complexity and polynomial quantum-query complexity.
\end{proof}

\subsection{Invariance under Negligible Deviations from Exact Reflections} \label{sec:relax_reflection}

In this formulation, each oracle-containing unit need only
be within negligible distance in operator norm from an exact
companion reflection. Under the requirements below, these
companions can themselves serve as the verifier's units, with only negligible error.
We therefore observe that this relaxation leaves the class $\QIMA^{\mathcal{O}}$
unchanged. 

\begin{definition}[$\QIMA^{\cal{O}}_{\rm{negl-refl}}$]
The class is defined identically to $\QIMA^{\cal{O}}$ as in Definition \ref{def:QIMA-O}, with the following difference. There exists a function $\delta(n)=\negl(n)$ such that each oracle-containing unit \(W_i^O\), making \(t_i\) quantum oracle queries, is accompanied by an exact reflection
circuit \(\widehat R_i^O\) such that
\[
(\widehat R_i^O)^\dagger=\widehat R_i^O,
\qquad
(\widehat R_i^O)^2=I,
\qquad
\left\|W_i^O-\widehat R_i^O\right\|_{\mathrm{op}}
\le \delta(n),
\]
and \(\widehat R_i^O\) makes at most \(t_i\) quantum oracle queries.
Additionally, on every promised oracle, the companion reflections
\(\{\widehat R_i^O\}_i\), together with all oracle-free commutation units,
form a pairwise-commuting family.
\end{definition}

\begin{remark} \label{remark:uniform_bound}
The same uniform classical preprocessing procedure must output
the circuit descriptions of both the actual units and their
companion reflections, with no additional classical oracle
queries. The approximation bound must hold uniformly over all promised oracle instances, with a single negligible function $\delta(n)$.
\end{remark}

More generally, for a prescribed function $\delta(n)\geq 0$,
let $\QIMA^{O}_{\delta\text{-near-refl}}$ denote the same model
with approximation bound $\delta(n)$, without requiring
$\delta$ to be negligible. Thus,
\[
    \QIMA^{O}_{\mathrm{negl\text{-}refl}}
    =
    \bigcup_{\delta\ \mathrm{negligible}}
    \QIMA^{O}_{\delta\text{-near-refl}}.
\]
Throughout, $\QIMA^{O}$ denotes the exact-reflection model
of Definition~\ref{def:QIMA-O}.

\begin{observation}[Invariance under negligible deviations]
For every classical oracle $O$,
\[
    \mathsf{QIMA}^O_{\text{negl-refl}}
    =
    \mathsf{QIMA}^O.
\]
\end{observation}

\begin{proof}

Fix a classical oracle $O$.
The containment
$\mathsf{QIMA}^O \subseteq
 \mathsf{QIMA}^O_{\text{negl-refl}}$
is immediate: take each oracle-containing unit to be its
own companion. For the reverse containment, we use the
companion reflections in place of the actual units.

Let \(\mathcal V^O\) be the actual $\QIMA^{\cal{O}}_{\rm{negl-refl}}$-verifier. Construct an ideal verifier
\(\widehat{\mathcal V}^O\) by replacing every oracle-containing unit
\(W_i^O\) by its companion exact reflection \(\widehat R_i^O\), while
leaving all oracle-free units unchanged. Let $T(n)$ be a polynomial upper bound on the total number
of quantum oracle queries in the actual units.

There exists a negligible function $\delta(n)$ such that for every \(i\),
\[
\left\|
\operatorname{ctrl}(W_i^O)
-
\operatorname{ctrl}(\widehat R_i^O)
\right\|_{\mathrm{op}}
=
\left\|W_i^O-\widehat R_i^O\right\|_{\mathrm{op}}
\le \delta(n).
\]
The corresponding unitary channels are therefore at diamond-norm distance
at most \(2\delta(n)\). A hybrid argument over the at most \(T(n)\) replaced
units shows that, for every promised oracle and every witness, the acceptance
probabilities of \(\mathcal V^O\) and \(\widehat{\mathcal V}^O\) differ by
at most
\[
\varepsilon(n):=2T(n)\delta(n).
\]

Let $c(n)$ and $s(n)$ be the actual verifier's completeness
and soundness, and write $g(n):=c(n)-s(n)$.
Since $T(n)$ is polynomial, $\delta(n)$ is negligible, and
$g(n)$ is inverse polynomial, we have
\[
    \varepsilon(n) \le \frac{g(n)}{4}
\]
for all sufficiently large $n$.
For these input lengths, the ideal verifier therefore has
completeness at least
\[
    \widehat c(n):=c(n)-\frac{g(n)}{4}
\]
and soundness at most
\[
    \widehat s(n):=s(n)+\frac{g(n)}{4}.
\]
These thresholds are efficiently computable and satisfy
\[
    \widehat c(n)-\widehat s(n)=\frac{g(n)}{2}.
\]

The ideal verifier satisfies Definition~1.4: its
oracle-containing units are exact reflections, its
oracle-free units may be arbitrary unitaries, and its
completed units commute on promised inputs.
The same preprocessing produces the companion circuits,
so no additional classical queries are needed.
Each companion also makes no more quantum queries than
the corresponding actual unit. Thus
$\mathsf{QIMA}^O_{\text{negl-refl}}
 \subseteq \mathsf{QIMA}^O$,
proving the equality.

\end{proof}

\subsection{$\BQP^{\cal{O}} \not\subseteq \QIMA^{\cal{O}}$}

We now explain how the query separation yields a single oracle separating the
two complexity classes. Let a disjoint slice of a classical oracle
$\mathcal O$ encode, for each even input length $n$, a promised Forrelation pair
$(f_n,g_n)$. Odd-length inputs are outside the promise. For each even $n$, the $n$-th slice consists of addresses $bx$,
where $b\in\{0,1\}$ and $x\in\{0,1\}^n$: addresses $0x$
encode $f_n(x)$, and addresses $1x$ encode $g_n(x)$.
Each such address therefore has length $n+1$. Define the associated promise problem by
\[
\begin{aligned}
1^n\in L_{\mathrm{yes}}^{\mathcal O}
&\quad\Longleftrightarrow\quad
\Phi(f_n,g_n)\geq\alpha,\\
1^n\in L_{\mathrm{no}}^{\mathcal O}
&\quad\Longleftrightarrow\quad
|\Phi(f_n,g_n)|\leq\beta.
\end{aligned}
\]
The
constant-query Forrelation circuit decides $L^{\mathcal O}$ with bounded error for every oracle
whose slices satisfy this promise, and hence
$L^{\mathcal O}\in\BQP^{\mathcal O}$.

To construct the separating oracle, enumerate all uniform
polynomial-time candidate verifiers, together with their claimed efficiently computable completeness
and soundness bounds $c_i(n)>s_i(n)$, as $(V_i,c_i,s_i)_{i\geq 1}$. We maintain a length bound $\ell_{i-1}$
such that all oracle values at addresses of length at most
$\ell_{i-1}$ have been permanently fixed.

At stage $i$, choose an even $n_i>\ell_{i-1}$ sufficiently large
that the total query bound of $V_i$ on $1^{n_i}$ is smaller than
the lower bound in Theorem~\ref{thm:qima_lower_bound}. Let $\ell_i\geq n_i+1$ bound the length of every classical or
quantum oracle query that $V_i$ can make on input $1^{n_i}$,
over all possible oracle answers. Temporarily fill all still-unassigned oracle addresses of length
at most $\ell_i$, except those in the $n_i$-th slice, as follows:
use a fixed perfect Forrelation pair for each unassigned
even-indexed slice, and set all odd-indexed slices to zero. Previously fixed values remain unchanged.

Now vary the $n_i$-th slice over all promised Forrelation pairs.
If some such pair causes $V_i$ to violate the structural requirements in Definition \ref{def:QIMA-O}, choose that pair. Otherwise, these
requirements hold throughout this family, and the proof of
Theorem~\ref{thm:qima_lower_bound} guarantees a pair on which $V_i$ violates its
claimed completeness or soundness bound. 

Permanently fix the chosen pair and all the temporary
assignments through length $\ell_i$. Later stages modify only
longer addresses, so neither the structural violation nor the
completeness or soundness failure can be undone. After all
stages, assign arbitrary values to any remaining addresses.

The resulting oracle defeats every candidate
$\QIMA^{\mathcal O}$ verifier at some input length, while the
constant-query Forrelation circuit continues to decide $L^{\mathcal O}$.
Therefore, there exists a classical oracle $\mathcal O$ such that
\[
\BQP^{\mathcal O}
\not\subseteq
\QIMA^{\mathcal O}.
\]

\subsection{One-Query Units with Restricted Address Support}
\label{subsec:restricted-address-support-lower-bound}

We extend the Forrelation lower bound to certain oracle-containing units that need
not be reflections. This extends the oracle separation from $\BQP^{\mathcal{O}}$, to this model as well.
In addition to the units already permitted in
$\mathrm{QIMA}^O$, allow units of the form
\[
    U_i^{f,g}=A_i O_{f,g}^{(Q_i)} B_i,
\]
where $A_i$ and $B_i$ are oracle-free circuits. The superscript $Q_i$
specifies the query register. All units must still commute on
promised oracles, and the verifier has no trusted workspace beyond the
fresh control qubits used for the Hadamard tests.

We set a $k\le n$ and denote it the \textit{support bound}, and require that, for each \(i\), \(A_i\) and \(B_i\) together act nontrivially on at most \(k\) address qubits of the query $O_{f,g}^{(Q_i)}$.\footnote{Address qubits in an oracle query are defined in Definition~\ref{def:oracle-access}
.} We next show that this restriction bounds the polynomial degree of the spectral projectors of each unit’s Hadamard-test acceptance operator. For a unitary $U$, this acceptance operator is
\[
    E(U)=\frac{2I+U+U^\dagger}{4}.
\]
For a Hermitian operator $H$ and a set $J\subseteq\mathbb R$, let
$\mathbf{1}_J(H)$ denote its spectral projector onto eigenvalues in $J$.\footnote{This is the orthogonal projector onto the sum of
the eigenspaces of $H$ whose eigenvalues lie in $J$.}

\begin{lemma}[Spectral projectors of a restricted one-query unit]\label{lem:restricted-query-spectral-degree}

Fix \(g\), a support bound $k$, and oracle-free circuits \(A,B\) that together act
nontrivially on at most \(k\) address qubits of the query in
\[
U_f = A O_{f,g} B.
\]
For every fixed \(\lambda \in \mathbb{R}\), each matrix entry of
\[
P_\lambda^f = \mathbf{1}_{\{\lambda\}}\!\left(E(U_f)\right)
\]
in the computational basis is a multilinear polynomial of degree
at most \(2^k\) in the truth-table values of \(f\).
\end{lemma}

\begin{proof}
Let \(S \subseteq [n]\) be the set of address positions on which
\(A\) or \(B\) acts nontrivially, so \(|S| \le k\). For each
assignment \(z\) to the remaining address positions $[n]\setminus S$, let \(\mathcal H_z\) be the subspace of the witness space in which the address qubits outside \(S\) are fixed to \(z\). Both \(A\) and \(B\) preserve \(\mathcal{H}_z\), since they
leave these qubits untouched. The phase oracle also preserves
\(\mathcal{H}_z\), since it is diagonal in the computational basis. Thus \(U_f\), \(E(U_f)\), and \(P_\lambda^f\) are all block diagonal
with respect to
\[
\mathcal{H} = \bigoplus_z \mathcal{H}_z.
\]
Within each block \(\mathcal{H}_z\), the address bits outside \(S\) are fixed,
so the oracle depends on at most \(2^{|S|}\le 2^k\) values of \(f\),
one for each assignment to the address bits in \(S\). Since \(g,A,B\) are fixed, these values of \(f\) completely determine
the block \(U_{f,z}\), hence \(E(U_{f,z})\) and its
\(\lambda\)-spectral projector. Each matrix entry of that projector
is therefore a function of at most \(2^k\) signs.

Every function \(F : \{-1,+1\}^d \to \mathbb{C}\) has the
multilinear representation
\[
F(s_1,\ldots,s_d)
=
\sum_{\sigma \in \{-1,+1\}^d}
F(\sigma)
\prod_{j=1}^d \frac{1+\sigma_j s_j}{2},
\]
whose degree is at most \(d\). Each computational-basis vector belongs to a single block, so every
matrix entry of \(P_\lambda^f\) is either an entry of one block's
spectral projector or zero if its row and column belong to different
blocks. Thus the degree bound for the individual blocks also holds
for every entry of \(P_\lambda^f\).
\end{proof}

\begin{theorem}[Forrelation lower bound with restricted one-query units]
\label{thm:restricted-query-lower-bound}
Consider the extended verifier model above. Let $C(n)$ bound the number
of classical preprocessing queries, and let $T_{\mathrm{refl}}(n)$ bound
the total number of queries in the oracle-containing reflection units.
Suppose the verifier also uses at most $m(n)$ one-query units, with
support bounds $k_1(n),\ldots,k_{m(n)}(n)$.

If the verifier decides Forrelation with parameters $(\alpha,\beta)$
and completeness and soundness $c(n)>s(n)$, then, for even $n$,
\begin{equation}
    C(n)+T_{\mathrm{refl}}(n)
    +\sum_{i=1}^{m(n)}2^{k_i(n)}
    \ge \beta 2^n-O(1).
    \label{eq:restricted-query-lower-bound}
\end{equation}
\end{theorem}

\begin{proof}
We follow the proof of Theorem~\ref{thm:qima_lower_bound}, modifying the construction of its
polynomial $p$ to account for the additional one-query units.
Fix even $n$, set $N=2^n$, and choose a perfect Forrelation pair
$(h,g_\star)$ with $h=H_Ng_\star$. As in that proof, let $\tau$ be the
preprocessing transcript on this pair, let $F$ be the set of queried
positions of $h$ in the preprocessing stage, and write $r=|F|\le C(n)$. Keeping $g=g_\star$ fixed,
restrict to functions $f$ agreeing with $h$ on $F$. Every such $f$
produces the same transcript, so all circuit descriptions are fixed.
We suppress the dependence on $n$. Let $\ell\le m(n)$ be the number
of additional one-query units in the circuit produced on transcript
$\tau$, indexed so that their address-support bounds are
$k_1,\ldots,k_\ell$.

Retain the notation $P^f$ for the product of the reflection acceptance
operators and $B$ for the product of the oracle-free acceptance
operators. Thus $B$ is fixed, with $0\le B\le I$. Writing
$E_i^f=E(U_i^f)$ for the acceptance operator of the  one-query units, the accepting
operator on promised oracles is now

\[
A^f = P^f B \prod_{i=1}^{\ell} E_i^f.
\]

For each $i$, let $\Lambda_i$ be the union of the spectra of $E_i^f$
as $f$ ranges over all functions agreeing with $h$ on $F$, with
$g=g_\star$.\footnote{Thus, $\lambda\in\Lambda_i$ if and only if $\lambda$ is an
eigenvalue of $E_i^f$ for at least one such function $f$.} The sets $\Lambda_i\subseteq[0,1]$ are finite and independent of the
particular $f$. Define
\[
    P_{i,\lambda}^f=1_{\{\lambda\}}(E_i^f).
\]

Set $\theta=(c+s)/2$. In place of the single threshold projector \(Q\) used in Theorem 3.5, we use a threshold adapted to each possible tuple of eigenvalues of the acceptance operators of the one-query units. On a joint eigenspace where \(E_i^f\) has eigenvalue \(\lambda_i\), these operators contribute the acceptance factor \(\prod_{i=1}^{\ell}\lambda_i\). Accordingly, for each \(\lambda\in\Lambda_1\times\cdots\times\Lambda_\ell\), define
\[
Q_\lambda
=
\mathbf{1}_{(\theta,\infty)}
\left(\left(\prod_{i=1}^{\ell}\lambda_i\right)B\right).
\]
Each \(Q_\lambda\) is independent of \(f\). We now set
\[
p(f)
=
\sum_{\lambda\in\Lambda_1\times\cdots\times\Lambda_\ell}
\operatorname{Tr}\left(
Q_\lambda P^f\prod_{i=1}^{\ell}P^f_{i,\lambda_i}
\right).
\]

On promised oracles, all the projectors inside each trace commute. The product of the \(P^f_{i,\lambda_i}\) restricts to the corresponding joint eigenspace, \(P^f\) enforces the reflection constraints, and \(Q_\lambda\) selects the part where total acceptance exceeds \(\theta\). The trace is therefore the dimension of this selected subspace. Summing over \(\lambda\), we see that \(p(f)\) counts, with multiplicity, the eigenvalues of \(A^f\) exceeding \(\theta\). Since \(s<\theta<c\), completeness and soundness give
\[
p(h)\ge 1,
\qquad
p(f)=0\quad\text{whenever }(f,g_\star)\in\mathrm{No}.
\]
The degree argument in Theorem~\ref{thm:qima_lower_bound} bounds the degree of each entry of
$P^f$ by $T_{\mathrm{refl}}$. By Lemma~\ref{lem:restricted-query-spectral-degree}, each entry of
$P_{i,\lambda}^f$ has degree at most $2^{k_i}$. Since the
$Q_{\boldsymbol{\lambda}}$ are fixed, multiplication, summation, and
taking the trace yield, after multilinearization,
\[
\deg p \le D := T_{\mathrm{refl}}
    + \sum_{i=1}^{\ell} 2^{k_i}.
\]
This polynomial representation holds for all functions agreeing with $h$ on $F$, with
$g=g_\star$.

The rest of the proof is unchanged compared to that of Theorem~\ref{thm:qima_lower_bound}. Apply the flip-variable substitution
and Hamming-layer symmetrization from Theorem~\ref{thm:qima_lower_bound} to obtain a univariate
polynomial $q$ with $\deg q\le D$ and $q(0)=p(h)\ge1$.
The Forrelation value on layer $j$ is again $1-2j/N$, so the same
root count gives at least $\beta N-r-O(1)$ distinct zeros of $q$.
Since $q$ is nonzero,
\[
    D\ge\beta N-r-O(1).
\]
Using $r\le C(n)$ and the stated worst-case resource bounds proves
the claimed inequality.
\end{proof}

\begin{corollary}[Consequences of a uniform address-support bound]
\label{cor:restricted-query-support}
Suppose every additional one-query unit has support bound at most
$k(n)$, and let $T(n)$ be the total number of quantum queries, including
those in the reflection units. Then
\[
    C(n)+2^{k(n)}T(n)\ge\beta2^n-O(1).
\]
In particular, if $C(n)$ is polynomial, then
\[
    T(n)=\Omega\bigl(2^{n-k(n)}\bigr).
\]
Consequently, no polynomial-query verifier in this extended model
decides Forrelation when
\[
    n-k(n)=\omega(\log n).
\]
For constant $k$, the query lower bound remains $\Omega(2^n)$;
for $k=O(\log n)$, it is $2^n/\operatorname{poly}(n)$.
\end{corollary}

\begin{proof}
Each additional unit makes one query, so
$T_{\mathrm{refl}}+\sum_i2^{k_i}\le 2^k T$.
The conclusions follow from
Theorem~\ref{thm:restricted-query-lower-bound}.
\end{proof}

\paragraph{Negligible implementation errors.}
The same conclusion holds if the tested units are negligible
approximations to ideal units of the form above, provided that the
ideal units commute on promised oracles and satisfy the same query and
address-support bounds. As in Remark~\ref{remark:uniform_bound}, the same uniform classical preprocessing procedure
must output both the tested and ideal circuit descriptions, without
additional classical oracle queries. The approximation bounds must
hold uniformly over all promised oracle instances. Replacing units at operator-norm distances
$\delta_j$ changes any witness's acceptance probability by at most
$2\sum_j\delta_j$. For polynomially many tests and an
inverse-polynomial completeness--soundness gap, negligible errors
therefore preserve a positive gap, and the exact lower bound applies.

\section{Tightness of the Computational Model} \label{section:comp_model}

In this section, we prove results which discuss natural relaxations of our $\QIMA^{\cal{O}}$ model that go beyond relaxing it to $\QIMA^{\cal{O}}_{\rm{negl-refl}}$. Our results show that these relaxations collapse the model to $\QMA^{\cal{O}}$. This serves to illustrate the tightness of our model with respect to these changes.

\subsection{Collapse Due to Relaxing the Reflection Requirement}\label{subsec:inverse-poly-reflection}

Section~\ref{sec:relax_reflection} shows that our oracle separation persists when
oracle-containing units are negligibly close to
query-preserving exact companion reflections.
We now show that allowing an inverse-polynomial deviation
recovers the full power of $\QMA^{O}$.

We rely on the general approach of the amplification algorithms by Marriott and Watrous, as described in Section \ref{subsec:mw-two-projector}, with some modifications. Let $V$ be a unitary QMA verifier acting on a witness
register $M$ and an $a$-qubit workspace register $A$, where $A$ is initialized to $\lvert 0^a\rangle$. Let
$\Pi_{\mathrm{acc}}$ project onto the verifier's accepting output subspace, and
define
\begin{equation}
    P
    :=
    I_M\otimes \lvert 0^a\rangle\!\langle 0^a\rvert_A,
    \qquad
    Q
    :=
    V^\dagger \Pi_{\mathrm{acc}}V.
    \label{eq:mw-projectors-section}
\end{equation}
Thus, $P$ projects onto properly initialized inputs, while $Q$ projects
onto inputs that $V$ maps to the accepting subspace. For every witness
$\lvert\psi\rangle\in M$,
\begin{equation}
    \langle\psi,0^a\rvert Q\lvert\psi,0^a\rangle
    =
    \Pr\!\left[V\text{ accepts }\lvert\psi\rangle\right].
\end{equation}

By Jordan's lemma, the Hilbert space decomposes into one- and
two-dimensional subspaces invariant under both $P$ and $Q$. On a
nontrivial two-dimensional block, let $\lambda
    \in (0,1)$ denote the corresponding eigenvalue of
$\left.PQP\right|_{\operatorname{im}P}$. In the QMA interpretation, $\lambda$ is the acceptance probability for the unit vector \(|p\rangle=|\psi\rangle_M|0^a\rangle_A\) spanning \(\operatorname{im}P\) within that block. Thus, \(\lambda=\langle p|Q|p\rangle\).

Standard Marriott--Watrous amplification alternates the projective
tests associated with $P$ and $Q$. Consider the reflections
\begin{equation}
    R_P:=2P-I,
    \qquad
    R_Q:=2Q-I.
\end{equation}
Writing $\lambda=\cos^2\varphi$, the product $R_QR_P$ has eigenvalues
$e^{\pm 2i\varphi}$ on the corresponding two-dimensional block. Thus,
the acceptance probability $\lambda$ is encoded in the rotation
generated by the two reflections. These reflections may also be viewed as full selective phase shifts:
\begin{equation}
    e^{i\pi P}=I-2P=-R_P,
    \qquad
    e^{-i\pi Q}=I-2Q=-R_Q.
\end{equation}

Here, consider the spectral decomposition of $P=SDS^\dag$, where $D$ is a diagonal matrix. Then $e^{i\pi P}=Se^{i\pi D}S^\dag$, where in $e^{i\pi D}$ each value $\lambda$ along the diagonal of $D$ becomes $e^{i\pi \lambda}$. Consequently,
\begin{equation}
    e^{-i\pi Q}e^{i\pi P}
    =
    R_QR_P.
    \label{eq:full-phase-mw}
\end{equation}

For our purposes, full phase shifts are too far from the identity.
We instead introduce the small-phase Marriott--Watrous core
\begin{equation}
    \Gamma_\theta
    :=
    e^{-i\theta Q}e^{i\theta P}.
    \label{eq:gamma-theta}
\end{equation}
Equation~\eqref{eq:full-phase-mw} shows that
$\Gamma_\pi=R_QR_P$, whereas for small $\theta$ the operator
$\Gamma_\theta$ is close to the identity.

The geometric intuition is the same as in Marriott–Watrous amplification. Within a block with acceptance probability \(\lambda\) close to \(1\), the one-dimensional ranges of \(P\) and \(Q\) nearly coincide. The phase shifts \(e^{i\theta P}\) and \(e^{-i\theta Q}\), acting on these respective subspaces, therefore nearly cancel. The remaining phase discrepancy becomes smaller as \(\lambda\) increases.

The operator $\Gamma_\theta$ alone, however, does not give a sound
test when the purported workspace register $A$ is supplied by Merlin. For
example, every state in $\ker P\cap\ker Q$ is fixed by
$\Gamma_\theta$ and would therefore pass its Hadamard test with
probability $1$, despite being improperly initialized.

To control these directions, we define
\begin{equation}
    W_\theta
    :=
    R_P\Gamma_\theta
    =
    R_Pe^{-i\theta Q}e^{i\theta P}.
    \label{eq:W-theta}
\end{equation}
The additional factor $R_P$ has two roles. First, since
$\Gamma_\theta$ is close to the identity, $W_\theta$ is close to the
exact reflection $R_P$. Second, $R_P$ acts as $+1$ on
$\operatorname{im}P$ and as $-1$ on $\ker P$, thereby moving
improperly initialized directions away from the accepting eigenvalue
$+1$.

Recall that the accepting POVM element of the Hadamard test for a
unitary $W$ is
\begin{equation}
    E(W)
    :=
    \frac{2I+W+W^\dagger}{4}.
\end{equation}
An eigenvector of $W$ with eigenvalue $e^{i\alpha}$ is accepted with
probability
\begin{equation}
    \frac{1+\cos\alpha}{2}
    =
    \cos^2\!\left(\frac{\alpha}{2}\right).
\end{equation}
The following lemma makes the preceding intuition precise. It shows
that the eigenphase of $W_\theta$ closest to zero is determined
monotonically by the original QMA acceptance probability, while
$W_\theta$ remains close to the reflection $R_P$.

\begin{lemma}[Small-phase Marriott--Watrous gadget]
\label{lem:small-phase-mw}
Let $P$ and $Q$ be orthogonal projectors, with
$\operatorname{im}P\neq\{0\}$, and define
\begin{equation}
    R_P:=2P-I.
\end{equation}
For $0<\theta\le 1/4$, let
\begin{equation}
    W_\theta
    :=
    R_Pe^{-i\theta Q}e^{i\theta P},
\end{equation}
and let
\begin{equation}
    E_\theta
    :=
    \frac{2I+W_\theta+W_\theta^\dagger}{4}
\end{equation}
be the accepting POVM element of its Hadamard test. Write
\begin{equation}
    \omega
    :=
    \lambda_{\max}
    \left(
        \left.PQP\right|_{\operatorname{im}P}
    \right).
\end{equation}
Then
\begin{equation}
    \lambda_{\max}(E_\theta)
    =
    G_\theta(\omega),
    \label{eq:G-optimum}
\end{equation}
where
\begin{equation}
    G_\theta(\lambda)
    :=
    \frac{
        1+\sqrt{1-(1-\lambda)^2\sin^2\theta}
    }{2}.
    \label{eq:G-definition}
\end{equation}
Moreover, $G_\theta$ is increasing on $[0,1]$, and
\begin{equation}
    G_\theta\!\left(\frac23\right)
    -
    G_\theta\!\left(\frac13\right)
    \ge
    \frac{\theta^2}{72}.
    \label{eq:G-gap}
\end{equation}
Finally,
\begin{equation}
    \|W_\theta-R_P\|_{\mathrm{op}}
    \le
    2\theta.
    \label{eq:W-close-RP}
\end{equation}
\end{lemma}

\begin{proof}
Apply Jordan's lemma to $P$ and $Q$. Since $W_\theta$ and $E_\theta$
are constructed from $P$ and $Q$, they preserve every Jordan block.

Consider first a two-dimensional block with parameter
$\lambda\in(0,1)$, and write
\begin{equation}
    s_\lambda
    :=
    \sqrt{\lambda(1-\lambda)}.
\end{equation}
In the basis supplied by Jordan's lemma,
\begin{equation}
    P
    =
    \begin{pmatrix}
        1&0\\
        0&0
    \end{pmatrix},
    \qquad
    Q
    =
    \begin{pmatrix}
        \lambda&s_\lambda\\
        s_\lambda&1-\lambda
    \end{pmatrix},
    \qquad
    R_P
    =
    \begin{pmatrix}
        1&0\\
        0&-1
    \end{pmatrix}.
\end{equation}
Using
\begin{equation}
    e^{-i\theta Q}
    =
    I+(e^{-i\theta}-1)Q,
    \qquad
    e^{i\theta P}
    =
    \begin{pmatrix}
        e^{i\theta}&0\\
        0&1
    \end{pmatrix},
\end{equation}
a direct multiplication gives
\begin{equation}
    \left.W_\theta\right|_{\mathcal H_\ell}
    =
    \begin{pmatrix}
        \lambda+(1-\lambda)e^{i\theta}
        &
        (e^{-i\theta}-1)s_\lambda
        \\[1mm]
        (e^{i\theta}-1)s_\lambda
        &
        -\lambda-(1-\lambda)e^{-i\theta}
    \end{pmatrix}.
\end{equation}
In particular,
\begin{equation}
    \det\!\left(
        \left.W_\theta\right|_{\mathcal H_\ell}
    \right)
    =
    -1,
\end{equation}
and
\begin{equation}
    \operatorname{Tr}\!\left(
        \left.W_\theta\right|_{\mathcal H_\ell}
    \right)
    =
    2i(1-\lambda)\sin\theta.
\end{equation}
Since $W_\theta$ is unitary, its two eigenvalues on this block can
therefore be written as
\begin{equation}
    e^{i\alpha_\lambda}
    \qquad\text{and}\qquad
    -e^{-i\alpha_\lambda},
\end{equation}
where $0\le \alpha_\lambda\le\theta$ and
\begin{equation}
    \sin\alpha_\lambda
    =
    (1-\lambda)\sin\theta.
    \label{eq:alpha-lambda}
\end{equation}

\noindent The corresponding eigenvalues of $E_\theta$ are
\begin{equation}
    \frac{1+\cos\alpha_\lambda}{2}
    \qquad\text{and}\qquad
    \frac{1-\cos\alpha_\lambda}{2}.
\end{equation}
Since $\alpha_\lambda\le\theta\le1/4$, the first is the larger one.
By~\eqref{eq:alpha-lambda}, it equals
\begin{equation}
    \frac{
        1+\sqrt{1-(1-\lambda)^2\sin^2\theta}
    }{2}
    =
    G_\theta(\lambda).
\end{equation}

\noindent It remains to consider the one-dimensional Jordan blocks. The four
possibilities give
\begin{center}
\begin{tabular}{c|c|c}
$P$ & $Q$ & eigenvalue of $E_\theta$\\
\hline
$1$ & $1$ & $1$\\
$1$ & $0$ & $\cos^2(\theta/2)$\\
$0$ & $1$ & $\sin^2(\theta/2)$\\
$0$ & $0$ & $0$
\end{tabular}
\end{center}
The first two values are $G_\theta(1)$ and $G_\theta(0)$,
respectively. Blocks contained in $\ker P$ have acceptance probability
at most $\sin^2(\theta/2)\le G_\theta(0)$. Hence maximizing over all
Jordan blocks gives~\eqref{eq:G-optimum}.

Differentiating~\eqref{eq:G-definition} gives
\begin{equation}
    G_\theta'(\lambda)
    =
    \frac{
        (1-\lambda)\sin^2\theta
    }{
        2\sqrt{1-(1-\lambda)^2\sin^2\theta}
    }
    \ge 0,
\end{equation}
so $G_\theta$ is increasing. For
$\lambda\in[1/3,2/3]$,
\begin{equation}
    G_\theta'(\lambda)
    \ge
    \frac{\sin^2\theta}{6}
    \ge
    \frac{\theta^2}{24},
\end{equation}
where the last inequality uses
$\sin\theta\ge\theta/2$ for $0\le\theta\le1/4$.
Integrating over $[1/3,2/3]$ yields
\begin{equation}
    G_\theta\!\left(\frac23\right)
    -
    G_\theta\!\left(\frac13\right)
    \ge
    \frac{\theta^2}{72},
\end{equation}
proving~\eqref{eq:G-gap}.

Finally, since $R_P$ is unitary,
\begin{align}
    \|W_\theta-R_P\|_{\mathrm{op}}
    &=
    \|e^{-i\theta Q}e^{i\theta P}-I\|_{\mathrm{op}}
    \\
    &\le
    \|e^{-i\theta Q}-I\|_{\mathrm{op}}
    +
    \|e^{i\theta P}-I\|_{\mathrm{op}}
    \\
    &\le
    4\sin\!\left(\frac{\theta}{2}\right)
    \\
    &\le
    2\theta.
\end{align}
This proves~\eqref{eq:W-close-RP}.
\end{proof}

Equation~\eqref{eq:alpha-lambda} gives the precise small-phase analogue
of the Marriott--Watrous rotation picture. On a block with QMA
acceptance probability $\lambda$, the two opposite phase shifts leave
a residual phase $\alpha_\lambda$ satisfying
\begin{equation}
    \sin\alpha_\lambda
    =
    (1-\lambda)\sin\theta.
\end{equation}
Thus, the residual phase decreases monotonically as the original
acceptance probability $\lambda$ increases, and vanishes when $\lambda=1$. On this two-dimensional block, \(W_\theta=R_P\Gamma_\theta\) is unitary with determinant \(-1\), so its two eigenvalues have the form
\[
e^{i\alpha_\lambda}
\qquad\text{and}\qquad
-e^{-i\alpha_\lambda}.
\]
Since \(W_\theta\) is close to \(R_P\), these lie near \(+1\) and \(-1\), respectively.
Hence one eigenvalue lies near $+1$ and carries the acceptance
information, while the other lies near $-1$. The Hadamard test
converts the first eigenphase into the monotone acceptance probability
$G_\theta(\lambda)$.

At the same time, $R_P$ suppresses the improperly initialized
one-dimensional blocks: blocks with $P=0$ have Hadamard-test
acceptance probability at most $\sin^2(\theta/2)$, whereas without
$R_P$ a block with $P=Q=0$ would be accepted with certainty.

We now use this gadget to show that an inverse-polynomial relaxation
of the reflection requirement collapses the model to
$\QMA^{\mathcal O}$.

\begin{theorem}[Collapse at inverse-polynomial tolerance]
\label{thm:inverse-poly-collapse}
Let $\delta(n)$ be an efficiently computable function satisfying
\begin{equation}
    \delta(n)
    \ge
    \frac{1}{\operatorname{poly}(n)}.
\end{equation}
Then
\begin{equation}
    \QIMA^{\mathcal O}_{\delta\text{-near-refl}}
    =
    \QMA^{\mathcal O}.
    \label{eq:inverse-poly-collapse}
\end{equation}
Moreover, the containment
\begin{equation}
    \QMA^{\mathcal O}
    \subseteq
    \QIMA^{\mathcal O}_{\delta\text{-near-refl}}
\end{equation}
can be realized using a single oracle-containing unit.
\end{theorem}

\begin{proof}
The containment
\begin{equation}
    \QIMA^{\mathcal O}_{\delta\text{-near-refl}}
    \subseteq
    \QMA^{\mathcal O}
\end{equation}
is immediate: a $\QMA^{\mathcal O}$ verifier can perform the
classical preprocessing and execute all controlled-unit tests of the
near-reflection verifier.

For the reverse containment, let
$L\in\QMA^{\mathcal O}$. We may assume that $L$ has a unitary verifier $V_x^{\mathcal O}$ whose
optimal acceptance probability $\omega_x$ satisfies
\begin{equation}
    x\in L_{\mathrm{yes}}
    \Longrightarrow
    \omega_x\ge\frac23,
    \qquad
    x\in L_{\mathrm{no}}
    \Longrightarrow
    \omega_x\le\frac13.
    \label{eq:qma-amplified-gap}
\end{equation}

Let $M$ be the witness register and $A$ the verifier's $a$-qubit
workspace register. In the near-reflection verifier, Merlin supplies
both $M$ and $A$. Define
\begin{equation}
    P
    :=
    I_M\otimes
    \lvert 0^a\rangle\!\langle 0^a\rvert_A,
\end{equation}
and
\begin{equation}
    Q
    :=
    (V_x^{\mathcal O})^\dagger
    \Pi_{\mathrm{acc}}
    V_x^{\mathcal O}.
\end{equation}
Then
\begin{equation}
    \omega_x
    =
    \lambda_{\max}
    \left(
        \left.PQP\right|_{\operatorname{im}P}
    \right).
\end{equation}

Choose
\begin{equation}
    \theta
    :=
    \min\left\{
        \frac{\delta(n)}{4},
        \frac14
    \right\},
\end{equation}
and define the single oracle-containing unit
\begin{equation}
    W_x^{\mathcal O}
    :=
    R_Pe^{-i\theta Q}e^{i\theta P},
    \qquad
    R_P:=2P-I.
    \label{eq:collapse-unit}
\end{equation}
We designate $R_P$ as its companion exact reflection. The unit $W_x^{\mathcal O}$ is efficiently implementable. Indeed,
\begin{equation}
    e^{-i\theta Q}
    =
    (V_x^{\mathcal O})^\dagger
    e^{-i\theta\Pi_{\mathrm{acc}}}
    V_x^{\mathcal O}.
\end{equation}
The middle operation applies a phase conditioned on the verifier's
accepting output. Similarly, $e^{i\theta P}$ and $R_P$ are
oracle-independent phases conditioned on the workspace register being
$\lvert 0^a\rangle$. Hence, if $V_x^{\mathcal O}$ makes $q(n)$ oracle
queries, then $W_x^{\mathcal O}$ makes $2q(n)$ oracle queries, while
its companion reflection $R_P$ makes none.

By Lemma~\ref{lem:small-phase-mw},
\begin{equation}
    \|W_x^{\mathcal O}-R_P\|_{\mathrm{op}}
    \le
    2\theta
    \le
    \frac{\delta(n)}{2}
    <
    \delta(n).
\end{equation}
Thus $W_x^{\mathcal O}$ satisfies the required near-reflection
condition. Since there is only one unit, the pairwise-commutation
requirement is vacuous.

The same lemma shows that the optimal acceptance probability of the
Hadamard test for $W_x^{\mathcal O}$, maximized over the entire
register $M\otimes A$, is $G_\theta(\omega_x)$.

Since $G_\theta$ is increasing,~\eqref{eq:qma-amplified-gap} implies
that the constructed verifier has completeness at least
\begin{equation}
    G_\theta\!\left(\frac23\right)
\end{equation}
and soundness at most
\begin{equation}
    G_\theta\!\left(\frac13\right).
\end{equation}
By Lemma~\ref{lem:small-phase-mw}, its completeness--soundness gap is
therefore at least
\begin{equation}
    G_\theta\!\left(\frac23\right)
    -
    G_\theta\!\left(\frac13\right)
    \ge
    \frac{\theta^2}{72}.
\end{equation}
Since $\delta(n)$ is at least inverse polynomial,
$\theta^2$ is also at least inverse polynomial. The constructed
near-reflection verifier therefore has an inverse-polynomial
completeness--soundness gap. Hence
\begin{equation}
    \QMA^{\mathcal O}
    \subseteq
    \QIMA^{\mathcal O}_{\delta\text{-near-refl}}.
\end{equation}
Together with the first containment, this proves
\eqref{eq:inverse-poly-collapse}.
\end{proof}

\begin{remark}[Finite universal gate sets]
\label{rem:finite-gate-set-near-reflection}
The proof above treats the selective phase shifts as exact gates.
Under a fixed finite universal gate set, let
$\widetilde W_x^{\mathcal O}$ approximate $W_x^{\mathcal O}$ so that
\begin{equation}
    \|\widetilde W_x^{\mathcal O}
      -W_x^{\mathcal O}\|_{\mathrm{op}}
    \le
    \varepsilon,
\end{equation}
where
\begin{equation}
    \varepsilon
    \le
    \min\left\{
        \frac{\delta(n)}{2},
        \frac{\theta^2}{144}
    \right\}.
\end{equation}
Such precision requires only polynomially many gates and does not
increase the number of oracle queries. By the triangle inequality,
\begin{equation}
    \|\widetilde W_x^{\mathcal O}-R_P\|_{\mathrm{op}}
    \le
    \varepsilon+2\theta
    \le
    \delta(n).
\end{equation}
Moreover, since
\begin{equation}
    E(W)
    =
    \frac{2I+W+W^\dagger}{4},
\end{equation}
we have
\begin{equation}
    \|E(\widetilde W_x^{\mathcal O})
      -E(W_x^{\mathcal O})\|_{\mathrm{op}}
    \le
    \frac{\varepsilon}{2}.
\end{equation}
Hence the completeness--soundness gap remains inverse polynomial.
\end{remark}

\subsection{One-query Collapse Without the Reflection Requirement}
\label{subsec:one-query-collapse}

Section~4.1 constructs a single verifier unit whose Hadamard test simulates a
$\mathsf{QMA}^O$ verifier. We now
show that the collapse persists with only one query per unit, provided the
reflection requirement is completely removed. 

We reuse the unit from the proof of Theorem~4.2 and apply a cyclic
version of the clock technique underlying the Feynman--Kitaev
circuit-to-Hamiltonian construction \cite{Kitaev2002}.
This technique encodes successive computational steps using a clock
register and is central to Kitaev's proof of QMA-completeness of
the local-Hamiltonian problem. The cyclic propagation unitary we
use also appears in \cite[Section~6]{JW07}. Here, arranging one
oracle query per clock step allows the entire propagation unitary
to use a single query, while preserving an inverse-polynomial
acceptance gap.
\paragraph{Query convention.}
For inputs of length $n$, fix a polynomial bound $p(n)$ on
the length of every oracle query. We use a common query
register $R$ containing a $p(n)$-qubit padded address, a
length $\ell$, and an enable bit $e$. When $e=1$, the oracle
applies the phase associated with the first $\ell$ address
bits. It acts as the identity when $e=0$. Below, $O$ denotes this
involution, and one application counts as one query.

We normalize the original QMA verifier to use this interface
before applying the construction of Section~4.1.
Oracle-free gates route each address into $R$, set its
length, and compute its control condition into $e$;
these operations are undone after the query.
All additional registers requiring initialization are
included in the workspace $A$, and the projector $P$
checks their initialization as well. In the final verifier,
Merlin supplies the entire register $A$.

\paragraph{Verifier model.} Let $\mathsf{QIMA}_{\mathrm{1q}}^O$ be the variant of $\mathsf{QIMA}^O$ in
which every oracle-containing unit makes a single query and has the form
$A O^{(R)} B$, where $R$ is its query register and $A,B$ are
oracle-free polynomial-size circuits on the full witness register.\footnote{Swap gates allow us to assume without loss of generality that all oracle queries are routed through a fixed
query register.}
Here $O^{(R)}$ applies $O$ on $R$ and acts as the identity on all
remaining qubits. Oracle-free units remain unrestricted polynomial-size unitaries. The
commutation and inverse-polynomial-gap requirements are unchanged.

As in Section~\ref{subsec:inverse-poly-reflection}, the Hadamard test on a unitary $U$ has acceptance operator $E(U)$; write
\[
    \mu(U):=\lambda_{\max}(E(U))
    =\lambda_{\max}\!\left(\frac{2I+U+U^\dagger}{4}\right)
\]
for its maximum acceptance probability over all states of the register $U$ acts on.

We apply the following lemma to the unit $W_x^O$ constructed
in Section~4.1, whose Hadamard test encodes the original QMA
verifier's acceptance probability while remaining sound when
Merlin supplies the workspace. This unit may make polynomially
many oracle queries; the lemma converts it into a one-query
unit while preserving an inverse-polynomial acceptance gap.

\begin{lemma}[A one-query unit with a cyclic clock]
\label{lem:one-query-clock}
Let $W^O$ be a unitary circuit making $q$ calls to the phase
oracle $O$, where $q$ is even.\footnote{This condition holds for
the unit constructed in Section~4.1, which uses the original
verifier and its inverse once each. It allows us to pad the
query count to a power of two by inserting cancelling pairs
of oracle calls.}
Let $L\geq\max\{2,q\}$ be a power of two.

There is a unitary $S^O$ acting on the input register of $W^O$
together with a $\log_2 L$-qubit register $C$, which makes
exactly one oracle query and satisfies
\begin{equation}
    \mu(S^O)=f_L\bigl(\mu(W^O)\bigr),
    \qquad
    f_L(a):=
    \frac{1+\cos\!\left(\frac{\arccos(2a-1)}{L}\right)}{2}.
    \label{eq:one-query-acceptance}
\end{equation}
For $0\leq s<c\leq1$, this transformation preserves the
acceptance gap up to a factor of $L^2$:
\begin{equation}
    f_L(c)-f_L(s)\geq\frac{c-s}{L^2}.
    \label{eq:one-query-gap}
\end{equation}
The circuit size of $S^O$ is polynomial in $L$ and the size
of $W^O$. Neither $C$ nor any workspace register needs to
be initialized.
\end{lemma}

\begin{proof}
Since $O^2=I$ and $q$ is even, we can insert cancelling pairs of queries
until there are exactly $L$. After routing all queries through a fixed query register $R$
and absorbing the swaps into the oracle-free circuits, write
\[
    W^O=F_L O^{(R)} F_{L-1}\cdots F_1 O^{(R)} F_0,
\]
where the $F_j$ are oracle free. The initial block $F_0$ can
be moved to the end by conjugation:
\[
    \widetilde W^O
    :=F_0 W^O F_0^\dagger
    =(F_0F_L)O^{(R)}F_{L-1}\cdots F_1O^{(R)}.
\]
Conjugation preserves the spectrum, so
$\mu(\widetilde W^O)=\mu(W^O)$. Renaming the oracle-free
blocks, we therefore have
\[
    \widetilde W^O
    =(G_{L-1}O^{(R)})\cdots(G_0O^{(R)}).
\]

where each $G_t$ is oracle free. We add a register $C$, which serves as the "clock", that records which of these $L$
steps to perform. Define the oracle-free unitary $A$ by
\[
    A\bigl(|t\rangle_C\otimes|\psi\rangle\bigr)
    =|t+1\bmod L\rangle_C\otimes G_t|\psi\rangle.
\]
Then $S^O=A(I_C\otimes O^{(R)})$ first makes one oracle query,
applies the appropriate $G_t$, and advances the clock. The
clock value $t$ determines which oracle-free circuit $G_t$
is applied afterward. To implement this selection, we
condition each elementary gate in the circuit for $G_t$
on $C$ being in state $|t\rangle$, and then increment $C$
modulo $L$.

Although $G_t$ may act on many qubits, each of its elementary
gates acts on only $O(1)$ qubits. Its clock-controlled version
therefore acts on $\log_2 L+O(1)$ qubits and can be implemented
without additional workspace using polynomially many gates
in $L$.\footnote{An arbitrary unitary on $k$ qubits admits
a decomposition into $2^{O(k)}$ elementary gates without
additional workspace \cite{Barenco_1995}. Here
$k=\log_2 L+O(1)$. The same bound applies to the clock
increment.}
Consequently, the entire construction has polynomial size
and requires no initialized ancillas. Finally, choosing $L$ to be a power of two ensures that the
$L$ positions in the cycle exhaust all computational-basis
states of $C$; there are no unused clock states.

Next, we determine the spectrum of $S^O$. Choose an orthonormal
eigenbasis $\{|\psi^{(j)}\rangle\}_j$ of $\widetilde W^O$, with
eigenvalues $e^{i\varphi_j}$. For each $j$, let
$|\psi_t^{(j)}\rangle$ be the state after the first $t$ circuit
steps:
\[
    |\psi_0^{(j)}\rangle=|\psi^{(j)}\rangle,
    \qquad
    |\psi_{t+1}^{(j)}\rangle
    =G_tO^{(R)}|\psi_t^{(j)}\rangle
    \quad(0\leq t<L).
\]
Since the full circuit is $\widetilde W^O$, we have
$|\psi_L^{(j)}\rangle=e^{i\varphi_j}|\psi_0^{(j)}\rangle$.
Thus $S^O$ shifts the $L$ states
$|t\rangle_C\otimes|\psi_t^{(j)}\rangle$ cyclically, with an
additional phase $e^{i\varphi_j}$ when the clock returns
from $L-1$ to $0$.

For each $\zeta$ satisfying $\zeta^L=e^{i\varphi_j}$, the state
\[
    \frac{1}{\sqrt L}\sum_{t=0}^{L-1}
    \zeta^{-t}|t\rangle_C\otimes|\psi_t^{(j)}\rangle
\]
is an eigenvector of $S^O$ with eigenvalue $\zeta$.
Indeed, shifting each term forward multiplies its coefficient
by $\zeta$, and the condition $\zeta^L=e^{i\varphi_j}$ ensures
the same at the wraparound step. The $L$ distinct roots
therefore give all eigenvalues on this $L$-dimensional subspace.

These subspaces cover the entire input space. For each fixed
$t$, the same unitary circuit prefix maps the basis
$\{|\psi^{(j)}\rangle\}_j$ to
$\{|\psi_t^{(j)}\rangle\}_j$, which is therefore also an
orthonormal basis. Together with the clock states, the vectors
$|t\rangle_C\otimes|\psi_t^{(j)}\rangle$, over all $t,j$,
form an orthonormal basis of the enlarged register.
Hence the full spectrum of $S^O$ consists exactly of the
$L$th roots of each eigenvalue of $\widetilde W^O$.

Having determined the spectrum of $S^O$, we now use it to
compute the maximum acceptance probability of its Hadamard
test and show that the acceptance gap decreases by at most
a factor of $L^2$. Choose $\varphi_j\in[-\pi,\pi]$. The corresponding eigenvalues
of $S^O$ are
\[
    \exp\!\left(\frac{i(\varphi_j+2\pi k)}{L}\right),
    \qquad k=0,\ldots,L-1.
\]
Recall that a Hadamard test accepts an eigenvector with
eigenvalue $e^{i\alpha}$ with probability $(1+\cos\alpha)/2$.
Thus, among these roots, acceptance is maximized by a root
closest to $1$. Since $|\varphi_j|\leq\pi$, the root
$e^{i\varphi_j/L}$ achieves this maximum, giving acceptance
probability
\[
    \frac{1+\cos(|\varphi_j|/L)}{2}.
\]
To express this in terms of the original acceptance
probability, write $a_j=(1+\cos\varphi_j)/2$.
Then $|\varphi_j|=\arccos(2a_j-1)$, so the displayed
probability is exactly $f_L(a_j)$.

It remains to maximize over $j$ and bound the acceptance gap.
For $0<a<1$, setting $x=\arccos(2a-1)$ gives
\[
    f_L'(a)=\frac{\sin(x/L)}{L\sin x}\geq\frac{1}{L^2},
\]
where the inequality follows from concavity of sine on
$[0,\pi]$. In particular, $f_L$ is increasing, so
\[
    \mu(S^O)
    =f_L\bigl(\mu(\widetilde W^O)\bigr)
    =f_L\bigl(\mu(W^O)\bigr).
\]
Finally, integrating the derivative bound yields
$f_L(c)-f_L(s)\geq(c-s)/L^2$ for $0<s<c<1$.
Continuity includes the endpoints $s=0$ and $c=1$.
\end{proof}

\begin{theorem}[One-query collapse]
\label{thm:one-query-collapse}
Under the query convention above,
\[
    \mathsf{QIMA}_{\mathrm{1q}}^O=\mathsf{QMA}^O.
\]
The containment from right to left uses a single commutation unit of the
form $A(I_C\otimes O)$.
\end{theorem}

\begin{proof}
As in Theorem~\ref{thm:inverse-poly-collapse}, a $\mathsf{QMA}^O$ verifier can execute the preprocessing
and all the Hadamard tests, giving
$\mathsf{QIMA}_{\mathrm{1q}}^O\subseteq\mathsf{QMA}^O$.

For the reverse containment, start with a $\mathsf{QMA}^O$ verifier and
follow the proof of Theorem~\ref{thm:inverse-poly-collapse} through the construction of $W_x^O$, now
fixing $\theta=1/4$. By Lemma~\ref{lem:small-phase-mw}, its Hadamard test has completeness at
least $c_0$ and soundness at most $s_0$, where
\[
    c_0:=G_{1/4}(2/3),
    \qquad
    s_0:=G_{1/4}(1/3),
    \qquad
    c_0-s_0\geq\frac{1}{1152}.
\]

The only remaining change is to reduce the query count of $W_x^O$.
As in Theorem~\ref{thm:inverse-poly-collapse}, it uses the original verifier and its inverse once
each, so its query count is even. Choose a polynomially bounded power of
two $L=L(|x|)$, at least $2$, that bounds this count uniformly over inputs
of length $|x|$. Lemma~\ref{lem:one-query-clock} gives a one-query unit
$S_x^O$ whose Hadamard test has completeness at least $f_L(c_0)$ and
soundness at most $f_L(s_0)$, with gap
\[
    f_L(c_0)-f_L(s_0)
    \geq\frac{c_0-s_0}{L^2}
    \geq\frac{1}{1152L^2}.
\]
The witness also supplies the clock: the lemma's spectral analysis already
maximizes over all states on the enlarged register. The verifier performs
just this one Hadamard test, so commutation is automatic. The gap remains
inverse polynomial, proving the reverse containment.
\end{proof}

\paragraph{Precision and support.}
The finite-gate-set argument of Remark~4.3 applies to the oracle-free
circuit $A$. If $\|\widetilde A-A\|_{\mathrm{op}}\leq\varepsilon$ and
$\widetilde S^O=\widetilde A(I_C\otimes O)$, then
\[
    \|E(\widetilde S^O)-E(S^O)\|_{\mathrm{op}}
    \leq\varepsilon/2.
\]
Taking $\varepsilon\leq1/(2304L^2)$ preserves an inverse-polynomial gap
without increasing the query count. The circuit $A$ may act on all
oracle-address qubits and polynomially many other witness qubits; no
depth restriction is imposed.
\subsection{Collapse due to Enabling Ancilla Space} \label{subsec:ancilla-collapse}

\paragraph{Trusted active ancillas.}
In standard QMA, allowing the verifier trusted ancilla qubits does not change the
class: the verifier may equivalently ask the witness to supply these registers and check
that they are initialized correctly before running the verification procedure.
For QIMA, however, such an initialization check need not commute with the
remaining verification operations. In fact, allowing even a single trusted qubit
raises the power of the model to all of QMA.

This phenomenon is closely related to the work of Nagaj, Hangleiter, Eisert, and
Schwarz~\cite{Nagaj_2021}, who introduced the \emph{Pinned Commuting Local Hamiltonian}
problem. In this problem, one is given a commuting local Hamiltonian acting on a
register $W$ together with one distinguished qubit $A$, and the ground-energy
question is restricted to states in which $A$ is fixed to a prescribed state.
They prove that the Pinned Commuting $3$-Local Hamiltonian problem is
QMA-complete~\cite[Theorem~2]{Nagaj_2021}. For clarity, let
\[
    \mathsf{QIMA}_{\mathrm{ta}(1)}
\]
denote the variant of $\QIMA$ in which the verifier has one additional trusted
qubit, initialized to $\ket{0}$, on which the commutation units may act.\footnote{For \(\mathrm{QIMA}_{\mathrm{ta}(1)}\), we allow efficiently computable
thresholds \(c(x)>s(x)\) depending on the full classical input, with
\(c(x)-s(x)\ge 1/p(|x|)\) for a fixed polynomial \(p\).
Containment in QMA still follows by standard threshold normalization
and amplification. Lemma~\ref{lem:pinned-clh-qima} and Theorem~\ref{thm:qima-trusted-ancilla} use this convention.} It is
distinct from the fresh control qubits used for the individual Hadamard tests. We first show that the pinned commuting-Hamiltonian problem can be verified in
$\mathsf{QIMA}_{\mathrm{ta}(1)}$, following the approach of \cite{Nagaj_2021}.

The main difficulty is that the Hamiltonian terms give rise naturally to commuting positive operators, whereas QIMA requires commuting reflections. We resolve this in two steps. First, we enlarge the witness space so that each positive operator is represented by a local projector. Second, we use the single trusted qubit to combine these projective tests with the necessary initialization checks while preserving commutativity. This yields a family of exact commuting reflections whose acceptance probability remains a monotone function of the original Hamiltonian test.

\begin{lemma}[Verification of pinned commuting Hamiltonians]
\label{lem:pinned-clh-qima}
Let
\[
    H=\sum_{i=1}^{m} H_i
\]
be a commuting $k$-local Hamiltonian acting on $W\otimes A$, where $A$ is a
single qubit,
\[
    0\preceq H_i\preceq I,
    \qquad
    [H_i,H_j]=0
\]
for all $i,j$, and $m=\poly(n)$.  Suppose that $A$ is required to be
initialized to $\ket{0}$, and that it is promised that one of the following two cases holds:
\begin{equation}
\label{eq:promise}
\begin{aligned}
\exists\,\ket{\psi}_W \quad &
\bra{\psi,0} H \ket{\psi,0} \le a,
&& \text{(YES)},\\
\forall\,\ket{\psi}_W \quad &
\bra{\psi,0} H \ket{\psi,0} \ge b,
&& \text{(NO)}.
\end{aligned}
\end{equation}
where $\Delta:=b-a\geq 1/\poly(n)$.  Then this promise problem belongs to
$\mathsf{QIMA}_{\mathrm{ta}(1)}$ under the inverse-polynomial-gap definition.
\end{lemma}

\begin{proof}
We may assume that $0\leq a<b\leq m$, since otherwise the promised
instances are trivial. Set
\[
    \gamma := \frac{\Delta}{2m^2},
\]
and, for every $i$, define
\[
    E_i := I-\gamma H_i.
\]
Since $\Delta\leq m$, we have $0\preceq E_i\preceq I$. Moreover, the
$E_i$ commute because the $H_i$ commute. Define
\[
    E:=\prod_{i=1}^m E_i
      =\prod_{i=1}^m (I-\gamma H_i).
\]

Since the $H_i$ commute, they may be simultaneously diagonalized. For
common eigenvalues $h_i\in[0,1]$,
\[
    1-\gamma\sum_i h_i
    \leq
    \prod_i(1-\gamma h_i)
    \leq
    1-\gamma\sum_i h_i
      +\gamma^2\sum_{i<j}h_i h_j
    \leq
    1-\gamma\sum_i h_i+\frac{\gamma^2m^2}{2}.
\]
Therefore, as an operator inequality,
\[
    I-\gamma H
    \preceq E
    \preceq I-\gamma H+\frac{\gamma^2m^2}{2}I.
\]
Denote\[
    \omega :=
    \lambda_{\max}\!\left(
       (I_W\otimes\langle 0|_A)\,
       E\,
       (I_W\otimes|0\rangle_A)
    \right).
\]
Then due to the promise \eqref{eq:promise},

\[
    \text{YES}\quad\Longrightarrow\quad
    \omega\geq c:=1-\gamma a,
\]
whereas
\[
    \text{NO}\quad\Longrightarrow\quad
    \omega\leq s:=1-\gamma b+\frac{\gamma^2m^2}{2}.
\]
Note that
\[
    c-s
    \geq
    \gamma\Delta-\frac{\gamma^2m^2}{2}
    =
    \frac{3\Delta^2}{8m^2},
\]
which is inverse polynomial.

It remains to realize this test using commuting reflections. For each $i$, let the witness supply an additional qubit $B_i$. Define the
projector $P_i$ on $\operatorname{supp}(H_i)\otimes B_i$ by
\[
P_i=
\begin{pmatrix}
E_i & \sqrt{E_i(I-E_i)}\\
\sqrt{E_i(I-E_i)} & I-E_i
\end{pmatrix},
\]
where the $2\times2$ block decomposition is with respect to the
computational basis $\{|0\rangle,|1\rangle\}$ of $B_i$, and each block is
an operator on $\operatorname{supp}(H_i)$.

A
direct calculation shows that $P_i^2=P_i=P_i^\dagger$. Moreover,
\[
\langle0|_{B_i}P_i|0\rangle_{B_i}=E_i.
\]

Furthermore, the projectors $P_i$ commute pairwise: their matrix entries
are functions of the commuting operators $H_i$, and the qubits $B_i$ are
distinct.

Let $A$ denote the qubit that is pinned to $|0\rangle$ in the Hamiltonian
instance. In the $\QIMA_{\mathrm{ta}(1)}$ verifier, $A$ is supplied by the witness. Let $C$ be the
verifier's single trusted active qubit, initialized to $|0\rangle$.

Define the initialization checks
\[
    Q_0:=|0\rangle\langle0|_A,
    \qquad
    Q_i:=|0\rangle\langle0|_{B_i}
    \quad (1\leq i\leq m).
\]
The projectors $Q_0,\ldots,Q_m$ commute pairwise. Write
\[
    \Pi_+ := |+\rangle\langle+|_C,
    \qquad
    \Pi_- := |-\rangle\langle-|_C.
\]
For $1\leq i\leq m$, define
\[
    R_i :=
    (2P_i-I)\otimes\Pi_+
    + I\otimes\Pi_-,
\]
and, for $0\leq j\leq m$, define
\[
    S_j :=
    I\otimes\Pi_+
    +(2Q_j-I)\otimes\Pi_-.
\]
The $\QIMA_{\mathrm{ta}(1)}$ verifier's units are these $\{R_i\}_{i=1}^m, \{S_j\}_{j=0}^m$. Every $R_i$ and $S_j$ is an exact reflection. The $R_i$ commute among
themselves, and the $S_j$ commute among themselves. Moreover, every
$R_i$ commutes with every $S_j$, since their nontrivial actions are
supported on the orthogonal subspaces $\Pi_+$ and $\Pi_-$ of $C$. Thus all of the verifier's units are pairwise-commuting reflections.

Let
\[
    P:=\prod_{i=1}^m P_i,
    \qquad
    Q:=\prod_{j=0}^m Q_j.
\]
The projector corresponding to acceptance of all Hadamard tests is
\[
    \prod_{i=1}^m\frac{I+R_i}{2}
    \prod_{j=0}^m\frac{I+S_j}{2}
    =
    P\otimes\Pi_+ + Q\otimes\Pi_-.
\]
Since the trusted qubit $C$ is initialized to $|0\rangle$ and
$|\langle 0|+\rangle|^2=|\langle0|-\rangle|^2=1/2$, the accepting
operator on the witness register is therefore
\[
    \frac{P+Q}{2}.
\]
We now compute its largest eigenvalue. By Jordan's lemma for the two
projectors $P$ and $Q$,
\[
    \lambda_{\max}\!\left(\frac{P+Q}{2}\right)
    =
    \frac{1+\|PQ\|_{\mathrm{op}}}{2}.
\]
Let
\[
J:\mathcal H_W\to
\mathcal H_W\otimes\mathcal H_A\otimes\mathcal H_B
\]
be the isometry defined by
\[
J|\psi\rangle
=
|\psi\rangle_W\otimes |0\rangle_A\otimes |0^m\rangle_B.
\]
Its image is exactly the subspace onto which $Q$ projects, and therefore

\[
JJ^\dagger
=
I_W\otimes |0\rangle\langle0|_A
\otimes |0^m\rangle\langle0^m|_B
=
Q.
\]

Hence
\[
    \|PQ\|^2
    =
    \|QPQ\|
    =
    \|J^\dagger P J\|.
\]
By the definition of $P_i$,
\[
    \langle 0^m|_B P|0^m\rangle_B
    =
    \prod_{i=1}^m E_i
    =
    E.
\]

Since $P$ is a projector and $Q=JJ^\dagger$,
\[
\|PQ\|_{\mathrm{op}}^2
=
\|QPQ\|_{\mathrm{op}}
=
\|J^\dagger P J\|_{\mathrm{op}}.
\]
Moreover, because
\[
\langle0|_{B_i}P_i|0\rangle_{B_i}=E_i,
\]
we have
\[
J^\dagger P J
=
(I_W\otimes\langle0|_A)\,
E\,
(I_W\otimes|0\rangle_A).
\]
By the definition of $\omega$, this gives
\[
\|PQ\|_{\mathrm{op}}^2=\omega.
\]
Thus, maximizing over all states supplied by Merlin,
\[
    \Pr[\text{accept}]
    =
    \lambda_{\max}\!\left(\frac{P+Q}{2}\right)
    =
    \frac{1+\|PQ\|}{2}
    =
    \frac{1+\sqrt{\omega}}{2}.
\]
Thus the optimal acceptance probability of the reflection verifier is
exactly $(1+\sqrt{\omega})/2$.

Since this function is increasing, the verifier has completeness at least
\[
    \widehat c:=\frac{1+\sqrt c}{2}
\]
and soundness at most
\[
    \widehat s:=\frac{1+\sqrt s}{2}.
\]
Finally,
\[
    \widehat c-\widehat s
    =
    \frac{c-s}{2(\sqrt c+\sqrt s)}
    \geq
    \frac{c-s}{4}
    \geq
    \frac{3\Delta^2}{32m^2},
\]
which is inverse polynomial.

Each $R_i$ acts only on the support of $H_i$, the qubit $B_i$, and the
trusted qubit $C$, while each $S_j$ is $2$-local. Hence the locality
increases by only a constant. The projectors $P_i$ are efficiently
constructible because each $H_i$ acts on a constant-dimensional
subsystem. This proves the claim.
\end{proof}

% \begin{proof}
% Set
% \[
%     \gamma:=\frac{\Delta}{2m^2},
% \]
% and, for every $i$, define the local unitary
% \[
%     U_i
%     :=
%     \exp\!\left(2i\arcsin\sqrt{\gamma H_i}\right).
% \]
% Each $U_i$ acts on the same qudits as $H_i$, and the $U_i$ commute because the
% $H_i$ commute.  The accepting POVM element of the Hadamard test for $U_i$ is
% \[
%     E_i
%     :=
%     \frac{2I+U_i+U_i^\dagger}{4}
%     =
%     I-\gamma H_i.
% \]
% Hence the probability that all Hadamard tests accept a state
% $\ket{\psi,0}$ is
% \[
%     \bra{\psi,0}E\ket{\psi,0},
%     \qquad
%     E:=\prod_{i=1}^{m}(I-\gamma H_i).
% \]

% Since the $H_i$ commute, they may be simultaneously diagonalized.  For common
% eigenvalues $h_i\in[0,1]$,
% \[
%     1-\gamma\sum_i h_i
%     \;\leq\;
%     \prod_i(1-\gamma h_i)
%     \;\leq\;
%     1-\gamma\sum_i h_i
%       +\gamma^2\sum_{i<j}h_i h_j
%     \;\leq\;
%     1-\gamma\sum_i h_i+\frac{\gamma^2m^2}{2}.
% \]
% Therefore, as an operator inequality,
% \[
%     I-\gamma H
%     \preceq
%     E
%     \preceq
%     I-\gamma H+\frac{\gamma^2m^2}{2}I.
% \]
% A YES instance consequently has completeness at least
% \[
%     c=1-\gamma a,
% \]
% whereas a NO instance has soundness at most
% \[
%     s=1-\gamma b+\frac{\gamma^2m^2}{2}.
% \]
% Thus
% \[
%     c-s
%     \geq
%     \gamma\Delta-\frac{\gamma^2m^2}{2}
%     =
%     \frac{3\Delta^2}{8m^2},
% \]
% which is inverse polynomial.  The units are efficiently constructible, since each $H_i$ acts on a constant-dimensional subsystem.
% \end{proof}

\begin{theorem}[One trusted qubit collapses QIMA to QMA]
\label{thm:qima-trusted-ancilla}
Under the inverse-polynomial-gap definition of QIMA,
\[
    \mathsf{QIMA}_{\mathrm{ta}(1)}=\mathsf{QMA}.
\]
\end{theorem}

\begin{proof}
The containment
\[
    \mathsf{QIMA}_{\mathrm{ta}(1)}\subseteq\mathsf{QMA}
\]
is immediate: a QMA verifier can initialize the trusted qubit itself and
execute the QIMA verification procedure.

For the reverse containment, we use the result of Nagaj, Hangleiter, Eisert,
and Schwarz~\cite[Theorem~2]{Nagaj_2021}, who prove that the Pinned Commuting
$3$-Local Hamiltonian problem with one pinned qubit is QMA-complete.  We briefly
recall the idea of their construction.

They start from a QMA-hard $2$-local Hamiltonian whose terms can be divided
into two families,
\[
    H=\sum_i A_i+\sum_j B_j,
\]
such that the $A_i$ commute pairwise and the $B_j$ commute pairwise, although
the two families need not commute with one another.  After adjoining one qubit
$A$, they replace the terms by
\[
    A_i' := A_i\otimes\ket{+}\!\bra{+}_A,
    \qquad
    B_j' := B_j\otimes\ket{-}\!\bra{-}_A.
\]
All of the resulting terms commute: terms within each family commute by
construction, while terms from different families have orthogonal support on
the additional qubit.  If that qubit is pinned to $\ket{0}$, then
\[
    \bra{\psi,0}
    \left(\sum_i A_i'+\sum_j B_j'\right)
    \ket{\psi,0}
    =
    \frac{1}{2}\bra{\psi}H\ket{\psi}.
\]
Thus the pinned ground-energy problem for the resulting commuting
$3$-local Hamiltonian reproduces, up to a factor of $1/2$, the original
QMA-hard local-Hamiltonian problem.  This is the source of the QMA-hardness in
their theorem.

After an affine normalization of the local terms, which preserves locality,
commutativity, and an inverse-polynomial promise gap, Lemma~\ref{lem:pinned-clh-qima} places this
QMA-hard pinned problem in $\QIMA_{\mathrm{ta}(1)}$. In the verifier from
that lemma, the qubit pinned in the Hamiltonian instance is supplied by
the witness, while a single additional qubit is used as the trusted active qubit;
all of the verifier's commutation units are exact reflections. Composing its QMA-hardness reduction with
that verifier gives
\[
    \mathsf{QMA}\subseteq\mathsf{QIMA}_{\mathrm{ta}(1)}.
\]
Together with the reverse containment, this proves the theorem.
\end{proof}

The oracle model admits an even more direct construction.  Unlike the
$k$-local units of QIMA, a unit of $\mathsf{QIMA}^{O}$ may be an arbitrary
polynomial-size circuit.  We can therefore encode an entire
$\mathsf{QMA}^{O}$ verification procedure into one exact
oracle-containing reflection. Let
\[
    \mathsf{QIMA}^{O,\mathrm{ta}(1)}
\]
denote the variant of $\mathsf{QIMA}^{O}$ with one trusted active qubit
initialized to $\ket{0}$.

\begin{theorem}[Trusted ancilla in the oracle model]
\label{thm:qima-oracle-trusted-ancilla}
For every classical oracle $O$,
\[
    \mathsf{QIMA}^{O,\mathrm{ta}(1)}=\mathsf{QMA}^{O}.
\]
Moreover, the containment from right to left can be realized using a single
oracle-containing unit, which is an exact reflection.
\end{theorem}

\begin{proof}
The containment
\[
\QIMA^{\mathcal O,\mathrm{ta}(1)}
\subseteq
\QMA^{\mathcal O}
\]
is immediate by direct simulation.

For the reverse containment, let \(L^{\mathcal O}\in\QMA^{\mathcal O}\),
and let \(V_x^{\mathcal O}\) be a unitary $\QMA^{\mathcal O}$-verifier for \(L^{\mathcal O}\).
Let \(M\) denote its witness register and let \(A\) denote its \(a\)-qubit
workspace register, which in the original verification procedure is
initialized to \(\ket{0^a}\).  Let \(\Pi_{\mathrm{acc}}\) be the projector
onto the accepting output subspace.

In the \(\QIMA^{\mathcal O,\mathrm{ta}(1)}\) protocol, the witness supplies
both \(M\) and \(A\).  The verifier has one trusted qubit \(C\), initialized to
\(\ket 0\), which will be used to check coherently that \(A\) is properly
initialized.  Define
\[
P_0
:=
I_M\otimes\ket{0^a}\!\bra{0^a}_A
\]
and the unitary
\[
F
:=
P_0\otimes X_C
+
(I-P_0)\otimes I_C.
\]
Thus \(F\) flips \(C\) precisely on the subspace in which the purported
workspace \(A\) equals \(\ket{0^a}\).  The operation \(F\) has a polynomial-size ancilla-free circuit implementation.\footnote{Concretely, it flips \(C\) if and only if every qubit of \(A\) is \(0\). This is a multi-controlled NOT with negative controls on the qubits of \(A\).}. 

\noindent Define
\[
\widetilde V_x^{\mathcal O}
:=
(V_x^{\mathcal O}\otimes I_C)F
\]
and
\[
\widetilde\Pi_{\mathrm{acc}}
:=
\Pi_{\mathrm{acc}}\otimes\ketbra{1}_C.
\]
Finally, let the verifier's single commutation unit be
\begin{equation}
R_x^{\mathcal O}
:=
(\widetilde V_x^{\mathcal O})^\dagger
\bigl(2\widetilde\Pi_{\mathrm{acc}}-I\bigr)
\widetilde V_x^{\mathcal O}.
\label{eq:trusted-ancilla-oracle-reflection}
\end{equation}
Since \(2\widetilde\Pi_{\mathrm{acc}}-I\) is a reflection,
\[
(R_x^{\mathcal O})^\dagger
=
R_x^{\mathcal O},
\qquad
(R_x^{\mathcal O})^2=I.
\]
Thus \(R_x^{\mathcal O}\) satisfies the exact-reflection requirement for
oracle-containing units.  There is only one unit, so the commutation
requirement is vacuous.

If \(V_x^{\mathcal O}\) makes \(q(n)\) quantum oracle queries, then
\(R_x^{\mathcal O}\) makes \(2q(n)\) queries: one execution of
\(V_x^{\mathcal O}\) and one execution of its inverse.  All remaining
operations are oracle free and have polynomial-size circuits.

The accepting POVM element of the Hadamard test for the reflection
\(R_x^{\mathcal O}\) is
\[
\frac{I+R_x^{\mathcal O}}{2}
=
(\widetilde V_x^{\mathcal O})^\dagger
\widetilde\Pi_{\mathrm{acc}}
\widetilde V_x^{\mathcal O}.
\]
Let \(\rho\) be an arbitrary state supplied by the witness on \(M\otimes A\).
Since \(C\) is initialized to \(\ket 0\), the acceptance probability is
\begin{align}
\Pr[+]
&=
\operatorname{Tr}\left[
\widetilde\Pi_{\mathrm{acc}}\,
\widetilde V_x^{\mathcal O}
\bigl(\rho\otimes\ketbra{0}_C\bigr)
(\widetilde V_x^{\mathcal O})^\dagger
\right]
\nonumber\\
&=
\operatorname{Tr}\left[
\Pi_{\mathrm{acc}}\,
V_x^{\mathcal O}
P_0\rho P_0
(V_x^{\mathcal O})^\dagger
\right].
\label{eq:trusted-ancilla-acceptance}
\end{align}
The second equality follows because the projector
\(\ketbra{1}_C\) retains precisely the component on which \(F\) found
\(A=\ket{0^a}\).  In particular, components in which the purported workspace
is improperly initialized, as well as coherences between the properly and
improperly initialized subspaces, cannot contribute to acceptance.

Maximizing~\eqref{eq:trusted-ancilla-acceptance} over all states \(\rho\)
gives
\[
\max_{\rho}\Pr[+]
=
\lambda_{\max}
\left(
P_0
(V_x^{\mathcal O})^\dagger
\Pi_{\mathrm{acc}}
V_x^{\mathcal O}
P_0
\right).
\]
Identifying \(\operatorname{im}P_0\) with the original witness register \(M\),
the right-hand side is exactly the optimal acceptance probability of the
original \(\QMA^{\mathcal O}\) verifier:
\[
\max_{\ket{\psi}_M}
\bra{\psi,0^a}
(V_x^{\mathcal O})^\dagger
\Pi_{\mathrm{acc}}
V_x^{\mathcal O}
\ket{\psi,0^a}.
\]
The construction therefore preserves both completeness and soundness
exactly.  Hence
\[
\QMA^{\mathcal O}
\subseteq
\QIMA^{\mathcal O,\mathrm{ta}(1)},
\]
which proves the claimed equality.
\end{proof}

\appendix
\section{The Commuting Local-Hamiltonian Problem is $\QIMA$-Complete} \label{app:QIMA-complete}

Bostanci and Hwang prove that the commuting local-projector problem is complete for $\QIMA$~\cite{bostanci2025commutinglocalhamiltonians2d}. This is a specific instance of the commuting local-Hamiltonian problem, where all Hamiltonian terms are exact projectors. The fact that the commuting local-projector problem is QIMA-hard, immediately implies that this is also the case for the commuting local-Hamiltonian problem. It is implicit in the work of Bostanci and Hwang that the commuting local-Hamiltonian problem is in $\QIMA$, using a reduction by Irani and
Jiang~\cite[Section~1.2.2 and Lemma~48]{irani2023commutinglocalhamiltonianproblem}. Below, we directly use this reduction to describe an explicit proof of this fact. 

\begin{theorem}[Commuting local Hamiltonian is complete for $\QIMA$]
Let
\[
H=\sum_{i=1}^{m}H_i
\]
be a commuting $k$-local Hamiltonian on constant-dimensional qudits, where
$m=\operatorname{poly}(n)$,
\[
0\leq H_i\leq I
\qquad\text{and}\qquad
[H_i,H_j]=0
\quad\text{for all }i,j.
\]
Given thresholds $\alpha<\beta$ satisfying
\[
\beta-\alpha\geq\frac{1}{\operatorname{poly}(n)},
\]
it is promised that either
\[
\lambda_{\min}(H)\leq\alpha
\qquad\text{or}\qquad
\lambda_{\min}(H)\geq\beta.
\]
Under the standard gate and representation convention in Remark~\ref{remark: gate-convention}, according to which the spectral projectors of a
constant-dimensional local term can be represented as local gates, this
problem belongs to $\QIMA_{O(k)}$. Consequently, the commuting
local-Hamiltonian problem is complete for $\QIMA$, up to a constant-factor increase in locality.
\end{theorem}

\begin{proof}
Bostanci and Hwang show that the commuting local-projector problem belongs to
$\QIMA$~\cite[Lemma~A.1]{bostanci2025commutinglocalhamiltonians2d}.
It therefore suffices to map the given instance to a commuting
local-projector Hamiltonian that has ground state energy $0$ in the \textsc{YES}
and at least $1$ in the \textsc{NO} case.

\noindent For each $i$, write each Hamiltonian term $H_i$ using its spectral decomposition
\[
H_i=
\sum_{a\in\mathcal A_i}
\lambda_{i,a}\Pi_{i,a},
\]
where $\mathcal A_i$ indexes the orthogonal eigenspaces of $H_i$, and $\Pi_{i,a}$ is a projector on the eigenspace which corresponds to the eigenvalue $\lambda_{i,a}$. Since $H_i$
acts on only $k=O(1)$ constant-dimensional qudits, the set $\mathcal A_i$
has constant size.

Because $H_i$ and $H_j$ are commuting Hermitian operators, they are
simultaneously diagonalizable. It follows that all of their spectral
projectors commute:
\[
[\Pi_{i,a},\Pi_{j,b}]=0
\qquad
\text{for all }i,j,a,b.
\]
The Irani--Jiang reduction asks the prover to choose a label
$a_i\in\mathcal A_i$ for every term and then replaces $H_i$ by
$I-\Pi_{i,a_i}$. For fixed labels, these complementary projectors $\{I-\Pi_{i,a_i}\}_{i=1}^m$ have a
common zero-energy state precisely when the selected eigenspaces have a
nonzero common intersection. We now show how these ideas enable constructing a commuting projector Hamiltonian which meets the conditions of the desired reduction.

Choose rational approximations
$\widetilde\lambda_{i,a}$ satisfying
\[
\left|
\widetilde\lambda_{i,a}-\lambda_{i,a}
\right|
\leq\eta,
\qquad
m\eta<\frac{\beta-\alpha}{8}.
\]
Since the needed precision is inverse polynomial and each $H_i$ has constant
dimension, these approximations require only polynomially many bits and can
be computed in polynomial time.

For each $i$, introduce a constant-size classical label register $L_i$ whose
computational-basis states are indexed by $\mathcal A_i$, and define
\[
K_i
:=
\sum_{a\in\mathcal A_i}
\ket{a}\!\bra{a}_{L_i}
\otimes
\bigl(I-\Pi_{i,a}\bigr).
\]
Each $K_i$ is a projector. Moreover, the projectors
$K_1,\ldots,K_m$ commute because the label registers are distinct and all
the spectral projectors $\Pi_{i,a}$ commute. The kernel of $K_i$ consists
of states for which the quantum register lies in the eigenspace of $H_i$
specified by the value of $L_i$.

The projectors $K_i$ enforce that the labels describe a common eigenstate
of the original terms $H_i$. To complete the reduction, we must also
ensure that the energy of this eigenstate is small. We do this by adding
local classical constraints that check the inequality
\begin{equation}
\label{eq:label_ineq}
    \sum_{i=1}^m \widetilde{\lambda}_{i,a_i}
    \leq \frac{\alpha+\beta}{2}.
\end{equation}
The midpoint threshold leaves room for the total approximation error
$m\eta$: since $m\eta <\frac{\beta-\alpha}{2}$, every common eigenstate of energy at most $\alpha$ satisfies
this inequality, while any common eigenstate satisfying it has energy
strictly below $\beta$.

Let $C$ be a polynomial-size
bounded-fan-in Boolean circuit that takes in as an input a label string $a=(a_1,\ldots,a_m)$ and accepts precisely when the inequality $\eqref{eq:label_ineq}$ holds.
We enforce the inequality by a standard encoding of a
classical circuit by local constraints. 
Introduce additional classical registers holding a proposed value for
every wire of this circuit. For a gate computing $y=g(x_1,\ldots,x_r)$,
add the projector
\[
    D_g :=
    \sum_{\substack{x\in\{0,1\}^r,\ y\in\{0,1\}\\ y\neq g(x)}}
    |x,y\rangle\langle x,y|.
\]
This projector penalizes precisely the assignments that violate the
gate relation. Since the fan-in $r$ is bounded, it acts on only
constantly many bits. Add analogous diagonal projectors requiring
valid label encodings, agreement between the input wires and the label
registers, the prescribed values of any constant wires, and an output
bit equal to $1$. Denote all these classical constraint projectors by
$D_1,\ldots,D_s$. For fixed labels $a=(a_1,\ldots,a_m)$, an assignment
to the wire registers satisfies every $D_j$ if and only if it describes
an accepting computation of $C$.
Define
\[
    K := \sum_{i=1}^m K_i + \sum_{j=1}^s D_j.
\]
The $D_j$ commute with one another because they are diagonal in the
computational basis of the classical registers. They also commute
with every $K_i$: each $D_j$ acts only on the classical registers
and is diagonal in that basis, while each $K_i$ is block diagonal
in the same basis and acts within each block only on the original
quantum register. Together with the pairwise commutation of the
$K_i$, this shows that $K$ is a commuting local-projector Hamiltonian.

The construction has polynomial size. Let $d$ be a fixed
upper bound on the dimension of each original qudit, and set
$\ell:=\lceil\log_2 d\rceil$. Encoding each original qudit
into $\ell$ qubits makes the support of each $H_i$ contain
at most $k\ell$ qubits. Moreover, $|\mathcal{A}_i|\leq d^k$,
so the label register $L_i$ requires at most $k\ell$ qubits.
Consequently, each $K_i$ acts on at most $2k\ell$ qubits.
The classical checks, including validity of label encodings,
can be implemented by bounded-fan-in Boolean circuits and
constant-locality diagonal constraint projectors.
Thus the resulting commuting local-projector Hamiltonian
has locality $O(k)$ on qubits, with the implicit constant
depending only on the fixed local dimension and the
chosen fan-in bound.

Each original qudit $q$ of dimension $d_q$ is represented by the
basis states $|0\rangle,\ldots,|d_q-1\rangle$ of an $\ell$-qubit
register. The remaining $2^\ell-d_q$ basis states do not represent
original qudit states. When expressing $K_i$ in the qubit encoding, we keep its original
matrix entries and add zero rows and columns for every basis state
in which one of the qudit registers $q$ on which it acts has a
value in $\{d_q,\ldots,2^\ell-1\}$.
For each qudit register, we also add the projector onto the
remaining basis states. These projectors commute with the encoded
constraints and ensure that every zero-energy state lies entirely
in the subspace representing the original qudits. The locality
remains $O(k)$.

We now identify exactly when $K$ has zero energy. This happens if and
only if there are valid labels $a_1,\ldots,a_m$ and a nonzero state
$|\psi\rangle$ such that
\[
    \Pi_{i,a_i}|\psi\rangle=|\psi\rangle
    \quad\text{for every }i,
    \qquad
    \sum_{i=1}^m \widetilde{\lambda}_{i,a_i}
    \leq \frac{\alpha+\beta}{2}.
\]
Indeed, given such labels and a state, adjoining the wire values of the
accepting computation gives a state annihilated by every $K_i$ and
$D_j$. Conversely, every term of $K$ is positive semidefinite and
preserves each computational-basis assignment to the classical
registers. A zero-energy state therefore has a nonzero component with
fixed labels and wire values that is itself annihilated by every term.
On this component, the $K_i$ impose the stated eigenspace conditions,
and the $D_j$ impose the energy inequality.

We can now verify the two promise cases. If
$\lambda_{\min}(H)\leq\alpha$, simultaneous diagonalization gives a
common eigenvector of the $H_i$ with labels $a_i$ satisfying
\[
    \sum_{i=1}^m\lambda_{i,a_i}
    =\lambda_{\min}(H)\leq\alpha.
\]
For these labels,
\[
    \sum_{i=1}^m\widetilde{\lambda}_{i,a_i}
    \leq\alpha+m\eta
    <\frac{\alpha+\beta}{2}.
\]
Hence the characterization above gives $\lambda_{\min}(K)=0$.

If instead $\lambda_{\min}(H)\geq\beta$, a zero-energy state of $K$
would, by the same characterization, yield a common eigenstate of the
$H_i$ whose energy satisfies
\[
    \sum_{i=1}^m\lambda_{i,a_i}
    \leq\sum_{i=1}^m\widetilde{\lambda}_{i,a_i}+m\eta
    \leq\frac{\alpha+\beta}{2}+m\eta
    <\beta,
\]
a contradiction. Thus $K$ has no zero-energy state in the NO case.
Since its terms are commuting projectors, every eigenvalue of $K$ is
a nonnegative integer, so $\lambda_{\min}(K)\geq1$.

We have obtained a polynomial-time reduction to a commuting
local-projector Hamiltonian with ground energy $0$ in the YES case and
at least $1$ in the NO case. This implies that the $k$-local commuting local-Hamiltonian problem is in $\QIMA_{O(k)}$, and thus the commuting local-Hamiltonian problem is in $\QIMA$. For hardness, Bostanci and Hwang show that the commuting local-projector
problem is $\QIMA$-hard
\cite[Lemma~A.2]{bostanci2025commutinglocalhamiltonians2d}. Since commuting
local-projector Hamiltonians are a special case of commuting local
Hamiltonians, the latter problem is also $\QIMA$-hard. Combining
hardness with the containment proved above establishes
$\QIMA$-completeness.
\end{proof}

\clearpage 

\section{$\BQP$ Algorithm for Forrelation}\label{sec:forrelation_in_bqp}

\paragraph{Forrelation $\in \BQP^{\mathcal O}$.}
Aaronson and Ambainis presented a $\BQP$ algorithm for Forrelation~\cite{aaronson2014forrelationproblemoptimallyseparates}. Consider the Forrelation instance $(f,g)$, write $N=2^n$, and let
\[
    O_f\ket{x}=f(x)\ket{x},
    \qquad
    O_g\ket{x}=g(x)\ket{x}
\]
be phase queries to the two functions. Starting from $\ket{0^n}$, apply
\[
    H_N,\quad O_g,\quad H_N,\quad O_f,\quad H_N,
\]
and then measure in the computational basis. The amplitude of the outcome $0^n$ is
\[
    \bra{0^n}H_NO_fH_NO_gH_N\ket{0^n}
    =
    \frac{1}{N}f^{\mathsf T}H_Ng
    =
    \Phi(f,g).
\]
Consequently, the probability of observing $0^n$ is $\Phi(f,g)^2$. This probability is at least $\alpha^2$ on yes-instances and at most $\beta^2$ on no-instances, so a constant number of repetitions decides the problem.

\clearpage 

\bibliographystyle{alpha}
\bibliography{references}

\end{document}